\documentclass[12pt]{article}

\usepackage{natbib}
\usepackage{amsmath,amssymb,mathtools,commath,comment,lmodern,dsfont}
\usepackage{float}
\counterwithin{equation}{section}
\usepackage[dvipsnames]{xcolor}

\usepackage{amssymb,amsmath,amsfonts,eurosym,geometry,ulem,graphicx,caption,color,setspace,sectsty,comment,footmisc,caption,natbib,pdflscape}
\usepackage{amsthm,subfigure,array,centernot,pgfplots}
\usepackage[listings,skins,breakable]{tcolorbox}
\usepackage[hidelinks]{hyperref}
\usepackage{tikz}
\usepackage{pifont}

\newtheorem{theorem}{Theorem}[section]
\newtheorem*{theorem*}{Theorem}
\newtheorem{corollary}{Corollary}[theorem]
\newtheorem{lemma}{Lemma}[theorem]

\newtheorem{assumption}{Assumption}[section]

\theoremstyle{definition}

\definecolor{UBCblue}{RGB}{0, 0, 95} 
\definecolor{ForestGreen}{RGB}{34, 139, 34}
\definecolor{a2red}{RGB}{192,0,0}
\definecolor{a3sand}{RGB}{191,144,0}
\definecolor{a4green}{RGB}{0,204,0}

\newcolumntype{L}[1]{>{\raggedright\let\newline\\arraybackslash\hspace{0pt}}m{#1}}
\newcolumntype{C}[1]{>{\centering\let\newline\\arraybackslash\hspace{0pt}}m{#1}}
\newcolumntype{R}[1]{>{\raggedleft\let\newline\\arraybackslash\hspace{0pt}}m{#1}}

\newcommand{\CI}{\mathrel{\perp\mspace{-10mu}\perp}}
\newcommand{\nCI}{\centernot{\CI}}

\newcommand{\E}[1]{\operatorname{\mathbb{E}}\left[#1\right]}
\newcommand{\Prp}[1]{\operatorname{\mathbb{P}}\left[#1\right]}

\newcommand{\Cov}[1]{\operatorname{Cov}\left[#1\right]}

\usetikzlibrary{shapes,decorations,arrows,calc,arrows.meta,fit,positioning}
\tikzset{
    Latex-Latex,auto,node distance =1 cm and 1 cm, thick,
    state/.style ={rounded rectangle, draw, minimum width = 0.7 cm, minimum height = 0.7 cm},
    half circle/.style={semicircle, draw, shape border rotate=90, anchor=chord center, minimum width=0.7 cm},
    point/.style = {circle, draw, inner sep=0.04cm,fill,node contents={}},
    bidirected/.style={Latex-Latex,dashed},
    el/.style = {inner sep=2pt, align=left, sloped},
    line/.style={draw, line width=1, -},  
    cross/.style={cross out, draw=black, minimum size=2*(#1-\pgflinewidth), inner sep=0pt, outer sep=0pt},
    cross/.default={1pt}
}

\pgfplotsset{ every non boxed x axis/.append style={x axis line style=-},
     every non boxed y axis/.append style={y axis line style=-}}

\begin{document}

\begin{titlepage}
\title{Instrumented Common Confounding}
\author{{\scshape Christian Tien} \thanks{\href{mailto:ct493@cam.ac.uk}{ct493@cam.ac.uk}; Faculty of Economics, University of Cambridge} }
\date{\today}
\maketitle

\begin{abstract}
\noindent 
Causal inference is difficult in the presence of unobserved confounders. We introduce the \emph{instrumented common confounding} (ICC) approach to (nonparametrically) identify average causal (structural) effects with instruments, which are exogenous only conditional on some unobserved \emph{common confounders}. The ICC approach is most useful in rich observational data with multiple sources of unobserved confounding, where instruments are at most exogenous conditional on some unobserved common confounders. Suitable examples of this setting are various identification problems in the social sciences, dynamic panels, and problems with multiple endogenous confounders.
The ICC identifying assumptions are closely related to those in mixture models, proximal learning and IV. Compared to mixture models \citep{bonhomme2016}, we require less conditionally independent variables and do not need to model the unobserved confounder. Compared to proximal learning \citep{cui2020}, we allow for non-common confounders, with respect to which the instruments are conditionally exogenous. Compared to IV \citep{newey2003}, we allow instruments to be exogenous conditional on some unobserved common confounders, for which a set of observed variables is complete. 
We prove point identification with outcome model and alternatively first stage restrictions. We provide a practical step-by-step guide to the ICC model assumptions and present the causal effect of education on income as a motivating example. 
\vspace{0in}

\bigskip
\end{abstract}
\setcounter{page}{0}
\thispagestyle{empty}
\end{titlepage}
\pagebreak \newpage

\section{Introduction} \label{sec:intro}

Causal inference in observational data with unobserved confounders is difficult. Researchers inevitably rely on some unverifiable assumptions, which invite criticism. A popular approach towards identification is the use of instrumental variables (IV). Famously, instruments need to satisfy a relevance condition and exclusion restriction. 
Finding excluded instruments is often difficult or impossible in applications. Our novel \emph{instrumented common confounding} (ICC) approach identifies average causal (structural) effects with instruments, which are excluded (and relevant) conditional on some \emph{unobserved} confounders. We call these unobserved confounders \emph{common}, because we assume their association with other observed variables. In economics, a well-known common confounder is ability (cognitive skill) in the education production function. An alternative title for this paper would be \emph{IV with mismeasured confounders}, where the measurements of the unobservable confounders may be economically meaningful variables with their own effect on the outcome. 

The here introduced ICC approach links instrumental variable (IV) and proximal learning (also negative control) methods. 
Compared to proximal learning \citep{cui2020}, we relax the conditional unconfoundedness assumption for treatment $A$. Instead, we require conditional relevance of action-aligned proxies $Z$ with respect to treatment $A$. With conditional relevance, the assumptions imposed on action-aligned proxies $Z$ start to resemble those in IV. The main difference compared to IV \citep{newey2003} is that exclusion and relevance of instruments $Z$ are no longer required to hold conditional only on observed confounders, but may hold conditional on (observed and) unobserved common confounders $U$. Thus, the information contained in $U$ differs from proximal learning to our common confounding setting: In proximal learning, treatment $A$ is exogenous conditional on $U$. In our setting, the instruments $Z$ are exogenous conditional on $U$. While in proximal learning $U$ contains every unobserved source of variation that made treatment $A$ endogenous, in our common confounding setting $U$ contains every unobserved source of variation that made instruments $Z$ endogenous. This may be a much more realistic identifying assumption, e.g. when treatment is chosen by heterogeneous economic agents.
The identifying assumptions of our ICC method are also related to mixture models \citep{bonhomme2016}, compared to which less than three conditionally independent measurements of the unobservable are needed.

Our novel method, instrumented common confounding (ICC), for which we prove and explain identification in detail, is not a panacea. It replaces some strong, untestable identifying assumptions by other such assumptions. While the relevance assumptions of ICC are testable, its conditional exclusion restriction remains untestable, except for over-identifying restrictions tests (as in IV). To shed light on these assumptions, estimation of the education production function is thoroughly discussed as a motivating example of the ICC approach.

In section \ref{sec:lit}, we briefly discuss the related IV and proximal learning literature. Section \ref{sec:setup} contains the model setup in detail. Our main contribution with new identification results in the common confounding model are in section \ref{sec:id-outcome} and \ref{sec:id-first}. In section \ref{sec:id-outcome} the outcome model is linearly separable in the disturbance, whereas in section \ref{sec:id-first} different first stage reduced form monotonicity restrictions are considered. Some numerical examples are included in section \ref{sec:num-examples}. We present a practical algorithm, and describe the returns to education and a health treatment with individual choice as examples of ICC models in section \ref{sec:examples}. We conclude in section \ref{sec:conclusion} and provide proofs in the appendix.

\section{Related Literature} \label{sec:lit}
Instrumented common confounding (ICC) bridges the proximal learning \citep{cui2020} and nonparametric IV \citep{newey2003, imbens2009} methods. Recent proximal learning literature \citep{miao2018} extends nonclassical measurement error models with mismeasured confounders \citep{mahajan2006, hu2008, kasahara2009, kuroki2014}. All measurement error models are characterised by independence conditions between some observed variables conditional on the unobserved variable. In this sense, measurement error models are a specific application of mixture models, for which identification results are similarly available \citep{hett2000, hall2003, allman2009, bonhomme2016}. These identification results in mixture models have one thing in common: They assume the independence of three observed variables conditional on the unobserved variable. Under this key assumptions, the entire model is nonparametrically identified in conjunction with completeness conditions \citep{bonhomme2016}, which impose richness requirements on the observed variables relative to the unobserved variable. 

The proximal learning literature focuses on the identification of average causal effects (ATE, ATT) in the presence of a common, unobserved confounder \citep{miao2018, cui2020, singh2020}. Instead of identification of the entire model, only a causal effect of a treatment on an outcome is identified, using proxies to instrument for each other and thus account for the unobserved confounder. This restricted focus enables the identification of average causal effects with weaker conditional independence assumptions on the proxies compared to traditional mixture models. Specifically, the model no longer needs to contain three conditionally independent measurements of the unobserved confounder. In instrumental common confounding, we retain the conditional independence assumptions of proximal learning, but drop unconfoundedness conditional on the unobservable in favour of a relevance requirement for the excluded action-aligned proxies, which become our instruments. 

The resemblance of proximal learning and IV is noteworthy. In proximal learning, proxies are used as instruments for each other to adjust for the confounding effect of the unobservable common confounder. Contrary to traditional nonparametric IV \citep{newey2003}, the proximal learning problem is not ill- but well-posed \citep{deaner2018}. In proximal learning, conditional moments of observed variables only are constructed to identify bridge functions of observed variables. Non-unique bridge function are no problem, because any valid bridge function can be used to point-identify the average causal effect of interest \citep{kallus2021}. To our knowledge, we are the first authors to leverage the similarities in identifying assumptions of IV and proximal learning for a novel identification approach. From the perspective of proximal learning, we add a conditional relevance requirement for action-aligned proxies $Z$ (our instruments) with respect to treatment $A$. With this strengthened relevance requirement, we allow for conditional confoundedness of the treatment $A$ due to non-common confounders. From the perspective of nonparametric IV, we allow the instrument $Z$ to satisfy exclusion conditional on some unobserved common confounder $U$. Then, we add conditionally independent, observable outcome-inducing proxies $W$ to the model, which are sufficiently relevant for the common confounder $U$. As their name suggests, unlike usual proxies the outcome-inducing proxies $W$ may be economically meaningful with their own direct effect on outcome $Y$.

\section{Setup}  \label{sec:setup}

A treatment (action) $A \in \mathcal{A}$ is discrete or continuous, with base measure $\mu_A$ of $\mathcal{A}$. The counterfactual $Y(a) \in \mathbb{R}$ would be observed if we could set $a \in \mathcal{A}$. $Y=Y(A)$ is the outcome of observed action $A$. For notational simplicity, conditioning on observed covariates $X \in \mathcal{X} \subseteq \mathbb{R}^{d_X}$ is not made explicit, but always possible.
The causal effect of interest $J$ is a function of counterfactuals with a contrast function $\pi: \mathcal{A} \to \mathbb{R}$.
\begin{align}
J &\coloneqq \E{\int_{\mathcal{A}} Y(a) \pi(a) d \mu_A(a)}. \label{eq:causal-effect}
\end{align}
Due to unmeasured confounders, which may be discrete, continuous or any mix, exchangeability is violated: $Y(a) \nCI A$.
$Z \in \mathcal{Z} \subseteq \mathbb{R}^{d_Z}$ is a vector of \emph{conditionally exogenous}, relevant instruments for treatment $A$. Instruments $Z$, which may be discrete, continuous or any mix, are exogenous conditional on a subset of unobserved \emph{common} confounders $U \in \mathcal{U}$. 
\begin{align*}
Y(a) \nCI Z  \text{ but } Y(a) \CI Z \ | \ U
\end{align*}
In proximal learning, the equivalent of our instruments $Z$ is called action-aligned proxies \citep{deaner2021many}. Just like the action-aligned proxies in proximal learning, these instruments satisfy an exclusion restriction with respect to the potential outcomes $Y(a)$ conditional on the common confounders $U$ (\ref{a:icc}.\ref{a:iv-exog}). However, unlike in proximal learning where action-aligned proxies may not directly affect treatment $A$ \citep{cui2020}, we use variation in $Z$ to instrument for treatment $A$. This instrumentation step requires relevance of instruments $Z$ for treatment $A$ conditional on common confounders $U$ (\ref{a:icc}.\ref{a:iv-complete}). 
Consequently, we prefer to call $Z$ conditionally exogenous instruments rather than action-aligned proxies with a relevance requirement, but either term would be equally valid.

Below, the set of conditional independence and relevance assumptions of the simple common confounding model are listed. Noticeably, the below assumptions only impose a stronger relevance requirement for instruments $Z$ (\ref{a:icc}.\ref{a:iv-complete}) compared to standard proximal learning. 

\begin{assumption}[Simple Common Confounding Model] \label{a:icc}

\begin{enumerate}
  \item SUTVA: $Y=Y(A, Z)$ and $W=W(A,Z)$. \label{a:sutva}
  \item Instruments \label{a:iv}
  \begin{enumerate}
    \item Exclusion: $Y(a,z) = Y(a) \CI (A, Z) \ | \ U \ \forall a \in \mathcal{A}$. \label{a:iv-exog}
    \item Relevance (completeness): \label{a:iv-complete} For any $g(A, U) \in L_2(A, U)$,  
    \begin{align}  \E{g(A, U) | Z} &= 0 \text{ only when } g(A, U) = 0.  \label{eq:iv-complete} \end{align}
  \end{enumerate}
  \item Outcome-inducing proxies \label{a:nc} 
  \begin{enumerate}
    \item Exclusion: $W(a, z) = W \CI (A, Z) \ | \ U$.  \label{a:nc-exog}
    \item Relevance (bridge function): \label{a:nc-relevance} There exists some function $h_0 \in L_2(A, W)$ such that 
    \begin{align} \E{h_0(A, W) | A, U} &= k_0(A, U) \label{eq:npiv-bridge-outcome}  \end{align} 
    almost surely, where $k_0 \in L_2(A, U)$ is a function of interest.
  \end{enumerate}
\end{enumerate}
\end{assumption}

Assumption \ref{a:icc}.\ref{a:sutva} is a stable unit treatment value assumption \citep{imbens2015}. It implies no interference across units and is not the focus of this work. 
Instruments $Z$ must be independent from the potential outcomes $Y(a)$ conditional on common confounders $U$, including no direct effect on outcomes other than through the treatment as stated in assumption \ref{a:icc}.\ref{a:iv-exog}. Hence, by definition the common confounders $U$ contain all unobservables conditional on which the instruments would satisfy an exclusion restriction with respect to the potential outcomes.
Instrument relevance is formulated as a completeness condition with respect to $(A, U)$ in assumption \ref{a:icc}.\ref{a:iv-complete}. This completeness requirement means that conditional on the common confounders, the remaining variation in $Z$ must still be sufficiently relevant for treatment $A$. In some sense, this assumption sounds very similar to the standard relevance requirement in IV conditional on observed confounders: The exogenous variation in the instruments must be sufficiently relevant for the treatment.


The outcome-inducing proxies $W \in \mathcal{W} \subseteq \mathbb{R}^{d_W}$ may be discrete, continuous or any mix. outcome-inducing proxies $W$ may directly affect $Y$, while independent from $(A, Z)$ conditional on $U$ (\ref{a:icc}.\ref{a:nc-exog}). Their richness requirement \ref{a:icc}.\ref{a:nc-relevance} with respect to $U$ is stated as the existence of a bridge function $h_0 \in L_2(A, W)$, whose expectation conditional on $(A, U)$ must equal a function of interest $k_0(A, U)$, which closely relates to the causal effect of interest $J$ (see section \ref{sec:id-outcome} and \ref{sec:id-first}). The function of interest $k_0$ is either some average structural function, or defined by moment restrictions as in assumption \ref{a:outcome-model}. Completeness of $W$ for $U$ (conditional on $A$) is sufficient for \ref{a:icc}.\ref{a:nc-relevance} and can thus be used alternatively to ensure relevance of $W$ for $U$ without reference to a specific function of interest $k_0$. The action bridge function $h$ will generally not be unique whenever $W$ carries more information than $U$ \citep{kallus2021}.

\begin{figure} 
\caption{Example DAG of an instrumented common confounding model}
\label{f:model}
\centering
\begin{minipage}{0.5 \textwidth}
\begin{tikzpicture}[node distance=1.5cm and 0.75cm]

    \node [state, dashed] (u) {$U$};
    \node [state, below left=1.5cm and 1.25cm of u] (a) {$A$};
    \node [state, below right=1.5cm and 1.25cm of u] (y) {$Y$};
    \node [draw, rounded rectangle, left=of a,  minimum width = 0.7 cm,  minimum height = 0.7 cm] (z) {$ \ \ \ Z \ \ \ $};
    \node [draw, rounded rectangle, right=of y, minimum width = 0.7 cm, minimum height = 0.7 cm] (w) {$W$};
    \node [state, dashed, below right=1.5cm and 1.25cm of a] (ut) {$\tilde{U}$};
    \node [cross=8pt, line width=4pt, a2red, below right=0.4cm and -0.4cm of z] (cross) {};
    \node [a2red, below right=0.6cm and -1.2cm of z] (a2) {\textbf{2a}};
    \node [a3sand, below right=-0.2cm and 0.2cm of z] (a3) {\textbf{2b}};
    \node [a4green, below right=0cm and 1cm of u] (a4) {\textbf{3b}};
    \draw[line, style=-latex] (a) edge (y);
    \draw[line, style=-latex, a3sand, line width=2] (z) edge (a);
    \draw[line, style=-latex] (w) edge (y);
    \draw[line, a2red, densely dotted, out=270, in=225] (z) edge (y);
    \draw[line, a2red, densely dotted, out=270, in=180] (z) edge (ut);
    \draw[line, style=-latex, dashed] (u) edge (a);
    \draw[line, style=-latex, dashed] (u) edge (y);
    \draw[line, style=-latex, dashed] (ut) edge (a);
    \draw[line, style=-latex, dashed] (ut) edge (y);
    \draw[line, style=-latex, dashed, out=180, in=90] (u) edge (z);
    \draw[line, style=-latex, a4green, line width=2, dashed, out=0, in=90] (u) edge (w);
\end{tikzpicture}
\end{minipage}
\hfill
\begin{minipage}{0.45\textwidth}
\begin{itemize}
  \item[$Y$: ] Outcome
  \item[$A$: ] Treatment
  \item[$Z$: ] Instrument
  \item[$W$: ] outcome-inducing proxy
  \item[$U$:] Common Confounder
  \item[$\tilde{U}$:] Non-Common Confounder
\end{itemize}
\end{minipage}
\end{figure}
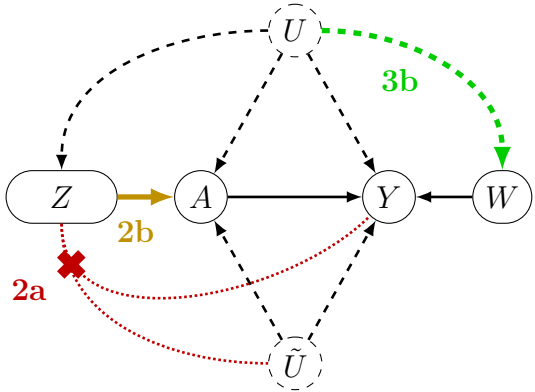

The directed acyclic graph (DAG) in figure \ref{f:model} is one of many possible representations of the conditional independences implied by assumption \ref{a:icc}. All unobserved variables and their direct effects are illustrated by dashed nodes and edges. 
The unobserved confounders $U$ are common, as they affect all observed variables $(Z, A, Y, W)$. At the same time, there are other unobserved confounders $\tilde{U}$ for the effect of $A$ on $Y$, with respect to which instruments $Z$ are exogenous (possibly conditional on $U$). The absence of edges between $(Y, \tilde{U})$ and $Z$ captures the exclusion condition satisfied by instruments $Z$ (\ref{a:icc}.\ref{a:iv-exog}): The instruments $Z$ satisfy exclusion conditional on the unobservable common confounders $U$. 
The relevance requirement for $Z$ is depicted with a thick directed edge from $Z$ to $A$ (\ref{a:icc}.\ref{a:iv-complete}). 
The outcome-inducing proxies $W$ have no direct edges to 
$(A, Z)$. Any association between $W$ and $(A, Z)$ stems from $U$, conditional on which they are independent, which reflects the exclusion restriction for $W$ (\ref{a:icc}.\ref{a:nc-exog}), even if it is not made as explicit in the graph as the exclusion restrictions for instruments $Z$. The richness requirement for outcome-inducing proxies $W$ with respect to $U$ (\ref{a:icc}.\ref{a:nc-relevance}) is illustrated with a thick arrow from $U$ to $W$. 
The confounder $\tilde{U}$ could be associated with $U$, but this link is omitted in favour of tractability in the DAG in figure \ref{f:model}.

\subsubsection*{Notation}
All results hold irrespective of whether we condition on covariates $X$, so $X$ is dropped from notation for simplicity.
$\mathbb{E}$ is the expectations operator wrt $(Y, A, Z, W)$. 
$\mathbb{E}_n$ is the empirical average over $n$ observations of $(Y, A, Z, W)$.
Let $L_2(O)$ be the space of square-integrable functions of a variable $O$ measurable wrt $(Y, A, Z, W)$.

\section{Identification with Outcome Model Restrictions}  \label{sec:id-outcome}
Identification of causal effect $J$ defined in \ref{eq:causal-effect} is the goal of this paper. 
In section \ref{sec:setup}, we introduced the common confounding model. In this section, we derive a main identification theorem, which relies on separability of the outcome model in the disturbance in assumption \ref{t:main-npiv-outcome} after introducing some useful lemmas. 

Using instruments for point identification of causal effects always requires some parametric model assumptions. One common approach is to formulate conditional moment restrictions for the outcome model. In assumption \ref{a:outcome-model}, we formulate such conditional moment restrictions. Treatment effect estimation with a continuous outcome fits into this framework. 

\begin{assumption}[Outcome model linearly separable in disturbance] \label{a:outcome-model} 
There exists some function $k_0 \in L_2(A, U)$ such that
\begin{align}
Y = Y(A) &= k_0(A, U) + \varepsilon, & \E{\varepsilon | Z, U} &= 0. \label{eq:outcome}
\end{align}
\end{assumption}

In the model described by assumption \ref{a:outcome-model}, the outcome $Y$ is linearly separable as a counterfactual mean function $k_0(A, U)$ and disturbance $\varepsilon$. A conditional moment holds, which states the mean-independence of disturbances $\varepsilon$ from instruments $Z$ given common confounders $U$. The treatments $A$ may therefore be endogenous.

\subsubsection*{IV with Fully Exogenous Instruments}
First, suppose the confounders $U$ were observed. Conditional on $U$, the instruments $Z$ satisfy an exclusion restriction. 
Then, the counterfactual mean function $k_0$ is identified by the completeness condition for $Z$ with respect to $A$ conditional on $U$ (\ref{a:icc}.\ref{a:iv-complete}). $k_0$ can be estimated by a variety of learning methods \citep{dikkala2020} from the conditional moment restrictions
\begin{align*}
\E{\left. Y - k_0(A, U) \right| Z, U} &= 0.
\end{align*}
The counterfactuals $Y(a)$ are closely related to the counterfactual mean function $k_0$. Once $k_0$ is identified, the causal effect $J$ is also identified. First, the contrast function $\pi$ is used to integrate out $A$ in $k_0(A, U)$ while $U$ is held fixed, producing $\phi_{IV}(U; k_0)$ as defined in lemma \ref{l:cf-outcome}. Then, the common confounders $U$ are integrated out from $\phi_{IV}(U; k_0)$ without dependence on treatment $A$ to obtain causal effect $J$.
\begin{lemma} \label{l:cf-outcome}
If assumption \ref{a:outcome-model} holds, then 
\begin{align*}
J = &\E{\phi_{IV}(U; k_0)} \\
\text{where} & \ \ \ \ \phi_{IV}(u; k_0) \coloneqq \int_{\mathcal{A}} k_0(a^\prime, u) \pi(a^\prime) d\mu_A(a^\prime) 
\end{align*}
\end{lemma}

When the common confounders $U$ are unobserved, conditioning on them is of course impossible. An alternative identification approach for $J$ is required. That alternative identification approach relies on bridge functions. First, we use that by assumption \ref{a:icc}.\ref{a:nc-relevance} there exists some bridge function $h_0 \in L_2(A, W)$, whose expectation conditional on $(A, U)$ equals the counterfactual mean function $k_0(A, U)$. In this sense, $h$ bridges the function spaces $L_2(A,U)$ and $L_2(A,W)$ for $k_0$. Let $\mathbb{H}_0$ be the nonempty set of valid outcome bridge functions $h$ defined by
\begin{align}
\mathbb{H}_0 &= \left\{ h \in L_2(A, W): \E{k_0(A, U) - h(A, W) | A, U} = 0 \right\} \neq \emptyset.
\end{align}
A sufficient condition for the existence of the outcome bridge function $h$ (assumption \ref{a:icc}.\ref{a:nc-relevance}) is that $W$ is complete with respect to $U$ conditional on $A$.
\begin{align*} \E{g(A,U)|A,W} = 0 \text{ only when } g(A,U)=0 \text{ for any } g \in L_2(A,U). \end{align*}
Completeness can be understood as a nonparametric relevance requirement. Intuitively, the richness of relevant variation in $W$ with respect to $U$ allows us to replicate any counterfactual mean $k_0(A, U)$ with the outcome bridge function $h_0(A, W)$ by taking the expectation of the latter conditional on $(A, U)$. Due to this replicability, the causal effect $J$ can be retrieved similarly as in lemma \ref{l:cf-outcome}. First, the contrast function $\pi$ is used to integrate out $A$ in a valid bridge function $h_0(A, W)$ while $W$ is held fixed, producing $\tilde{\phi}_{IV}(W; h_0)$ as defined in lemma \ref{l:h-outcome}. Then, the outcome-inducing proxies $W$ are integrated out from $\tilde{\phi}_{IV}(W; h_0)$ without dependence on treatment $A$ to obtain causal effect $J$. Lemma \ref{l:h-outcome} describes this identification mechanism formally.
\begin{lemma}  \label{l:h-outcome}
Suppose \ref{a:outcome-model} and \ref{a:icc}.\ref{a:nc} hold. For any $h_0 \in \mathbb{H}_0$, 
\begin{align*}
J = &\E{\tilde{\phi}_{IV}(w; h_0)} \\
\text{where} & \ \ \ \ \tilde{\phi}_{IV}(w; h_0) \coloneqq \int_{\mathcal{A}} h_0(a, w) \pi(a) d\mu_A(a) 
\end{align*}
\end{lemma}

While lemma \ref{l:h-outcome} shifts the focus from identifying a function of partly unobservable $k_0 \in L_2(A, U)$ to a function of only observables $h_0 \in L_2(A, W)$, the moment conditions defining the valid set of outcome bridge functions $\mathbb{H}_0$ are still conditional on partly unobservable $(A, U)$. Hence, observed data cannot immediately identify any $h_0 \in \mathbb{H}_0$. 
The next step in the direction of identification is lemma \ref{l:obs-npiv-outcome}. The lemma states that any bridge function in the set $\mathbb{H}_0$ also satisfies equation \ref{eq:obs-mom-outcome}, which is a conditional moment restriction involving only observed variables. In the conditional moment restriction \ref{eq:obs-mom-outcome} we condition on the observable instruments $Z$ instead of partly unobservable $(A, U)$. 
\begin{lemma} \label{l:obs-npiv-outcome}
Under assumptions \ref{a:outcome-model} and \ref{a:icc}.\ref{a:nc}, any $h_0 \in \mathbb{H}_0$ satisfies that
\begin{align}
\E{Y - h_0(A, W) | Z} &= 0 \label{eq:obs-mom-outcome} 
\end{align}
\end{lemma}
These conditional moment restrictions in lemma \ref{l:obs-npiv-outcome} resemble those of IV. A main difference to IV with a conditional moment restriction of this type is that the instruments $Z$ do not need to be relevant with respect to both $A$ and $W$, which would also imply the uniqueness of $h_0$. Instead, $h_0$ will often be non-unique (specifically if $W$ contains more information than $U$), yet the causal effect $J$ is still point-identified.
The conditional moment restriction \ref{eq:obs-mom-outcome} defines a different set of observable bridge functions $\mathbb{H}^\text{obs}_0$, defined below.
\begin{align}
\mathbb{H}^\text{obs}_0 &= \{ h \in L_2(A, W): \E{Y - h(A, W) | Z} = 0 \} \neq \emptyset.
\end{align}
To be able to use the observable bridge functions in the identification of $J$, the remaining task is to relate $\mathbb{H}_0^{\text{obs}}$ and $\mathbb{H}_0$. Only then, we can hope to identify $J$ with some $h \in \mathbb{H}_0^{\text{obs}}$. Fortunately, the completeness condition \ref{a:icc}.\ref{a:iv-complete} for $Z$ with respect to $A$ conditional on $U$ ensures the equivalence of $\mathbb{H}_0^{\text{obs}}$ and $\mathbb{H}_0$. Lemma \ref{l:h-equal-outcome} states this equivalence result formally. 

\begin{lemma} \label{l:h-equal-outcome}
Under assumption \ref{a:icc},
\begin{align}
\mathbb{H}_0 = \mathbb{H}_0^{\text{obs}}.
\end{align}
\end{lemma}

Now that the equivalence of the set of original bridge functions $\mathbb{H}_0$ and observable bridge functions $\mathbb{H}_0^{\text{obs}}$ has been shown, our main theorem \ref{t:main-npiv-outcome} in the outcome model restrictions case follows straightforwardly. The causal effect $J$ is identifiable with any observable bridge function $h_0 \in \mathbb{H}_0^{\text{obs}}$. First, the contrast function $\pi$ is used to integrate out $A$ in an observable and valid bridge function $h_0(A, W)$, while $W$ is held fixed, producing $\tilde{\phi}_{IV}(W; h_0)$. Then, the outcome-inducing proxies $W$ are integrated out from $\tilde{\phi}_{IV}(W; h_0)$ without dependence on treatment $A$ to obtain causal effect $J$.


\begin{theorem}  \label{t:main-npiv-outcome}
Suppose \ref{a:outcome-model} and \ref{a:icc} hold. For any $h_0 \in \mathbb{H}_0^{\text{obs}}$, 
\begin{align*}
J = &\E{\tilde{\phi}_{IV}(w; h_0)} \\
\text{where} & \ \ \ \ \tilde{\phi}_{IV}(w; h_0) =  \int_{\mathcal{A}} h_0(a, w) \pi(a) d\mu_A(a) 
\end{align*}
\end{theorem}

\subsubsection*{Comment on Non-uniqueness}
The non-uniqueness of $h_0$ and point identification of $J$ may at first seem at odds. The observable bridge function $h_0 \in \mathbb{H}_0^{\text{obs}}$ will not be unique whenever $W$ contains more information than necessary for completeness with respect to $U$. Hence, the non-uniqueness of $h_0$ will only be reflected in its second argument, the outcome-inducing proxies $W$. While the $\tilde{\phi}_{IV}(W; h_0)$ may be non-unique as a function of $W$, the causal effect $J$ is unique once we take the expectation over $W$ in $\E{\tilde{\phi}_{IV}(W; h_0)}$ without dependence on treatment $A$, similar to the arguments in standard proximal learning \citep{kallus2021}. 


\section{Identification with First Stage Restrictions} \label{sec:id-first}

To obtain point identification in IV, first stage monotonicity assumptions represent an alternative to linear separability in the outcome model disturbance. The same logic applies to the common confounding model, but some additional steps are needed, which we lay out in this section. Some form of first stage strict monotonicity assumption could suit identification problems with continuous treatments.

We distinguish three main cases of the control function approach: First, when the common confounders $U$ are observed. Second, when the common confounders $U$ do not enter the first stage reduced form. Third, when the common confounders $U$ enter the first stage reduced form, but in a way that restricts the complexity of interactions.

\subsection{All Confounders Observed} \label{ssec:cffei}
The causal effect $J$ is identified by a simple control function strategy when the common confounders $U$ are observed. 

\begin{assumption}[Strict Monotonicity with Observed Confounders] \label{a:strict-monotonicity-simple} 
\begin{align}
Y &= g(A, U, \varepsilon) \label{eq:outcome-form-simple} \\
A &= h(Z, U, \eta) \label{eq:reduced-form-simple}
\end{align}
\begin{enumerate}
\item The reduced form $h(Z, U, r)$ is strictly monotonic in $r$ with probability 1, \label{a:mon-1-simple}
\item and $\eta$ is a continuously distributed scalar with a CDF that is strictly increasing on the support of $\eta$ (conditional on $U$). \label{a:mon-con-simple}
\item $(\eta, \varepsilon) \CI Z \ | \ U$. \label{a:mon-ind-simple} 
\end{enumerate}
\end{assumption}

In case of observed $U$, assumption \ref{a:strict-monotonicity-simple}.\ref{a:mon-ind-simple} states that $\varepsilon$ and $Z$ are independent conditional on $U$. By theorem 1 in \cite{imbens2009}, $A$ and $\varepsilon$ are then independent conditional on the control variable 
\begin{align}
V \coloneqq F_{A | Z, U}(A, Z, U) = F_{\eta | U}(\eta). \label{eq:control-var}
\end{align}
$V$ is a one-to-one function of $\eta$ due to the strict monotonicity of $h$ in its third argument and the strictly increasing CDF $F_{\eta | U}$ (assumption \ref{a:strict-monotonicity-simple}.\ref{a:mon-1-simple}/\ref{a:mon-con-simple}). For this reason, conditioning on $(V, U)$ is equivalent to conditioning on $(\eta, U)$. Then, as conditional on $(V, U)$ all variation in treatment $A$ stems from the instruments $Z$, $A$ and $\varepsilon$ are independent conditional on the control variable $(V, U)$ due to assumption \ref{a:strict-monotonicity-simple}.\ref{a:mon-ind-simple}.

Instrument relevance is expressed as a common support assumption \ref{a:common-support-simple} in this model.
\begin{assumption}[Common Support with Observed Confounders]   \label{a:common-support-simple}
For all $A \in \mathcal{A}$ where $\pi(A) \neq 0$, the support of $V$ conditional on $(A, U)$ equals the support of $V$ conditional on $U$.
\end{assumption}
A typical function of interest is the average structural function 
\begin{align}
k_0(A, U) &= \int_{\mathcal{V}} g(A, U, e) \dif F_{\varepsilon | U}(e),  \label{eq:asf-au}
\end{align}
which is identified by the standard nonparametric control function approach when the common confounders $U$ are observed \citep{imbens2009}. However, moving forward $U$ will be unobserved, so we instead focus on identification of the causal effect $J$ as defined in \ref{eq:causal-effect}. While the average structural function in \ref{eq:asf-au} is a function of $U$, the causal effect $J$ in \ref{eq:causal-effect} averages out the common confounders $U$ without dependence on the treatment $A$. $J$ is a less informative version of the average structural function $k_0$, because in $J$ there is one more level of averaging compared to $k_0$.

Let $k_{0, v_{(j)}}(a, v_{(j)}, u)$ be the condtional expectation $k_{0, v_{(j)}}(a, v_{(j)}, u) \coloneqq \E{Y | A=a, V_{(j)}=v_{(j)}, U=u}$, where $(j)$ may index different control variables $V_{(j)}$.
The generalised propensity score is the conditional density of $A$ given $V$ and $U$, $f_{A | V, U}(a | v, u)$, relative to base measure $\mu_A$ \citep{hirano2004}. 
Lemma \ref{l:simple-cf} below states that the IPW estimator $\phi_{IPW}$, regression estimator $\phi_{REG}$ and doubly robust estimator $\phi_{DR}$ allow identification of the causal effect $J$ conditional on the control variable $V$.

\begin{lemma} \label{l:simple-cf}
If assumptions \ref{a:strict-monotonicity-simple} 
and \ref{a:common-support-simple} hold, then 
\begin{align*}
J = &\E{\phi_{IPW}(Y, A, V, U; f_{A | V, U})} = \E{\phi_{REG}(V, U; k_{0, v})} = \E{\phi_{DR}(Y, A, V, U; f_{A | V, U}, k_{0, v})} \\
\text{where} & \ \ \ \ \phi_{IPW}(y, a, v, u; f_{A | V, U}) = y \frac{\pi(a)}{f_{A | V, U}(a|v,u)} \\
& \ \ \ \ \phi_{REG}(v, u; k_{0, v}) = \int_{\mathcal{A}} k_{0, v}(a^\prime, v, u) \pi(a^\prime) d\mu_A(a^\prime) \\
& \ \ \ \ \phi_{DR}(y, a, v, u; f_{A | V, U}, k_{0, v}) = (y - k_{0, v}(a, v, u)) \frac{\pi(a)}{f_{A | V, U}(a|v,u)} + \int_{\mathcal{A}} k_{0, v}(a^\prime, v, u) \pi(a^\prime) d\mu_A(a^\prime) 
\end{align*}
\end{lemma}

This baseline result is helpful, as we will find equivalents to it when the common confounders $U$ are unobserved.

\subsection{Common Confounders Do Not Affect the First Stage}  \label{ssec:cfccnifs}
If the common confounders $U$ are absent from the first stage reduced form, including the independence of $Z$ and $\eta$ unconditional on $U$, identification is again possible without further conditional independence assumptions. This model is described in assumption \ref{a:strict-monotonicity-no-cf-fs}. 

\begin{assumption}[Strict Monotonicity without Common Confounders in First Stage] \label{a:strict-monotonicity-no-cf-fs} 
\begin{align}
Y &= g(A, U, \varepsilon) \label{eq:outcome-form-no-cf-fs} \\
A &= h(Z, \eta) \label{eq:reduced-form-no-cf-fs}
\end{align}
\begin{enumerate}
\item The reduced form $h(Z, r)$ is strictly monotonic in $r$ with probability 1, \label{a:mon-1-no-cf-fs}
\item and $\eta$ is a continuously distributed scalar with a CDF that is strictly increasing on the support of $\eta$. \label{a:mon-con-no-cf-fs}
\item $\eta \CI Z$, and $\varepsilon \CI Z \ | \ U$. \label{a:mon-ind-no-cf-fs} 
\end{enumerate}
\end{assumption}
Assumption \ref{a:strict-monotonicity-no-cf-fs} primarily restricts the complexity of the first stage reduced form. The absence of $U$ from the reduced form does not imply that $U$ cannot affect treatment $A$. Below we provide an example of a simple model in which the common confounder $U$ is linearly associated with treatment $A$ and instruments $Z$. 
\begin{align*}
A &= Z \zeta + U \gamma_A + \epsilon_A,  & U &= Z \tilde{\gamma}_Z + \epsilon_U,  & (\epsilon_A, \epsilon_U) &\CI Z \\ \implies A &= Z \left( \zeta + \tilde{\gamma}_Z \gamma_A \right) + \eta,  & \eta &= \epsilon_A + \epsilon_U \gamma_A,  & \eta &\CI Z
\end{align*}
Assuming independence of instruments $Z$ from the structural disturbances $(\epsilon_A, \epsilon_U)$ suffices to write a reduced form, in which the instruments are independent from the disturbance $\eta$. 



Due to the absence of $U$ from the first stage reduced form, the variation in treatment $A$ can be separated into variation induced by instruments $Z$ versus disturbance $\eta$. Subject to assumption \ref{a:strict-monotonicity-no-cf-fs}, the control function
\begin{align*}
V_{\ref{a:strict-monotonicity-no-cf-fs}} &\coloneqq F_{A | Z}(A, Z) = F_\eta(\eta)
\end{align*}
is identifiable. 
$V_{\ref{a:strict-monotonicity-no-cf-fs}}$ is a one-to-one function of $\eta$ due to the strict monotonicity of $h$ in its second argument and the strictly increasing CDF $F_{\eta}$ (assumption \ref{a:strict-monotonicity-no-cf-fs}.\ref{a:mon-1-no-cf-fs}/\ref{a:mon-con-no-cf-fs}). For this reason, conditioning on $(V_{\ref{a:strict-monotonicity-no-cf-fs}}, U)$ is equivalent to conditioning on $(\eta, U)$. Then, as conditional on $(V_{\ref{a:strict-monotonicity-no-cf-fs}}, U)$ all variation in treatment $A$ stems from the instruments $Z$, $A$ and $\varepsilon$ are independent conditional on the control variable $(V_{\ref{a:strict-monotonicity-no-cf-fs}}, U)$ due to assumption \ref{a:strict-monotonicity-no-cf-fs}.\ref{a:mon-ind-no-cf-fs}. A result of the simple one-to-one relation of $V_{\ref{a:strict-monotonicity-no-cf-fs}}$ and $\eta$, the results of section \ref{ssec:bridge-functions} apply regarding the identification of causal effect $J$ with bridge functions.

\subsection{Control Bridge Function}  \label{ssec:control-bridge}
Another, more complicated, option to identify the causal effect $J$ relies on identification of a surrogate control function, which would render treatment $A$ exogenous if the common confounders $U$ were observed. First, this model is formally described in assumption \ref{a:strict-monotonicity}.

\begin{assumption}[Strict Monotonicity] \label{a:strict-monotonicity} 
\begin{align}
Y &= g(A, U, \varepsilon) \label{eq:outcome-form} \\
A &= h(Z, m(U, \eta)) \label{eq:reduced-form}
\end{align}
\begin{enumerate}
\item The reduced form $h(Z, m)$ is strictly monotonic in $m$ with probability 1, \label{a:mon-1}
\item $m(U, t)$ for $m: \mathcal{U} \times \mathcal{H} \rightarrow \mathbb{R}$ is strictly monotonic in $t$ with probability 1, \label{a:mon-2}
\item and $\eta$ is a continuously distributed scalar with a CDF that is strictly increasing on the support of $\eta$ (conditional on $U_0$). \label{a:mon-con}
\item $(\eta, \varepsilon) \CI Z \ | \ U$. \label{a:mon-ind} 
\end{enumerate}
\end{assumption}

The above strict monotonicity assumption \ref{a:strict-monotonicity} reduces the amount of interactions between $Z$ and $U$ in the reduced form for $A$ (\ref{a:strict-monotonicity}.\ref{a:mon-2}) on top of imposing strict monotonicity in the scalar disturbance $\eta \in \mathcal{H}$ (\ref{a:strict-monotonicity}.\ref{a:mon-1}/\ref{a:mon-2}). The scalar $\eta$ is continuously distributed and independent from $Z$ conditional on $U$, as is $\varepsilon$. The reduction in the dimension of interactions in the first stage reduced form is needed for the identifiability of some surrogate version of scalar $\eta$ when the common confounders are unobserved.

To identify some surrogate version of $\eta$ in the first stage reduced form, we need additional information in form of conditional independence assumptions. The requirements formulated in assumption \ref{a:icc-first-add} enable identification of the first stage to the degree that some useful control function for the endogenous variation $\eta$ becomes obtainable. These assumptions go beyond the basic common confounding model \ref{a:icc}. This approach can accommodate the fact that only a subset of common confounders $U_0 \in \mathcal{U}_0 \subseteq \mathcal{U}$ may enter the first stage reduced form in a complexity-restricted way. In fact, the relevance requirements in assumption \ref{a:icc-first-add} then are with respect $U_0$ instead of the more complex $U$. Our notation in this section sticks with $U$ throughout, but a researcher may consider less complex $U_0$ in the first stage to motivate the relevance requirements in assumption \ref{a:icc-first-add}.

\begin{assumption}[Additional First Stage Requirements] \label{a:icc-first-add}
\begin{enumerate}
\item SUTVA: $A(Z, W_0) = A$, $W_1(A, Z) = W_1$.
\item First stage instrument-aligned proxy: $A(z, w_0) = A(z) \CI W_0 \ | \ U$ \label{a:f-nci}
\item First stage action-inducing proxy: $W_1(a, z) = W_1 \CI Z \ | \ U$ \label{a:f-nca}
\end{enumerate}
\end{assumption}

Assumption \ref{a:icc-first-add}.\ref{a:f-nci} and \ref{a:icc-first-add}.\ref{a:f-nca} are typical proximal learning conditional independence assumptions applied to the first stage reduced form instead of the outcome equation. $W_0$ and $W_1$ are variables, which satisfy conditional independence assumptions with respect to $Z$ and $A$. For example, $W_0$ and $W_1$ could be different parts of a vector-valued outcome-inducing proxy $W$. 

Neither the conditional mean $k_{0, v}(a, v, u)$ nor the propensity score $f(a | v, u)$ are directly identifiable without $U$.
We use variation in the first stage proxies $W_0$ and $W_1$ to identify a \emph{control bridge function}.
We define the control quantity of interest $V_{\ref{a:strict-monotonicity}}$ as
\begin{align}
V_{\ref{a:strict-monotonicity}} &\coloneqq \int_{\mathcal{U}} F_{A | Z, U}(A | Z, u) \dif F(u).   \label{eq:tildev}
\end{align}
With the control quantity $V_{\ref{a:strict-monotonicity}}$ we can identify the causal effect of interest conditional on $U$. In section \ref{sssec:control-bridge-properties}, we prove and discuss this claim. In section \ref{sssec:control-bridge-id}, we show how to identify the control quantity from observed data.

\subsubsection{Properties of the Control Bridge} \label{sssec:control-bridge-properties}
The control quantity $V_{\ref{a:strict-monotonicity}}$ is useful, because it captures information about $\eta$, which would allow the exact identification of $\eta$ conditional on $U$. The strict monotonicity properties of the reduced form for treatment $A$ (\ref{a:strict-monotonicity}) allow us to write the control quantity as a function of $\eta$ and $U$. Let $m^{-1}(., U)$ be the inverse of $m(U, t)$ in $t$ and $h^{-1}(.,Z)$ the inverse of $h(Z, m)$ in $m$. Below, we show that the control quantity is equal to the expected conditional density of $\eta$, integrating out $U$ without dependence on instruments $Z$.
\begin{align*}
V_{\ref{a:strict-monotonicity}} &= \int_{\mathcal{U}} F_{A | Z, U}(A, Z, u) \dif F(u) \\
&= \int_{\mathcal{U}} \Pr{ \left( h(Z, m(u, \eta)) \leq A \right)} \dif F(u) \\
&= \int_{\mathcal{U}} \Pr{ \left( \eta \leq m^{-1} \left(h^{-1}(A, Z), u \right) \right) } \dif F(u) \\
&= \int_{\mathcal{U}} F_{\eta | U} \left( m^{-1} \left(h^{-1}(A, Z), u \right) \right)  \dif F(u) \\
&= \int_{\mathcal{U}} F_{\eta | U} \left( \eta(u) \right)  \dif F(u)
\end{align*}
The last line follows from monotonicity of the reduced form $h(Z, m(U, \eta))$ for treatment $A$ in $\eta$ and the continuous conditional distribution of $\eta$ (assumption \ref{a:strict-monotonicity}). The continuous conditional distribution of $\eta$ (\ref{a:strict-monotonicity}) implies that $\eta$ is exactly identified from $(A, Z, U)$, because it corresponds to a unique $F_{A | Z, U}$. As $\eta$ is identified only up to $U$ when $U$ is unobserved, we write $\eta(U) \coloneqq m^{-1} \left(h^{-1}(A, Z), U \right)$.

With the complexity restriction on the interactions of $Z$ and $(U, \eta)$ imposed by assumption \ref{a:strict-monotonicity}, $(V_{\ref{a:strict-monotonicity}}, U)$ contains the same information as $(V, U)$ (and hence $(\eta, U)$), which is formulated in lemma \ref{l:same-info} using sigma algebras.
\begin{lemma} \label{l:same-info}
Suppose assumption \ref{a:strict-monotonicity} holds. Then, the sigma algebras associated with the following three vectors of random variables are identical:
$$(V_{\ref{a:strict-monotonicity}}, U), \ (V, U), \ (\eta, U).$$
\end{lemma}
We include graphical intuition for the result that $\eta$ is identified exactly from the mean of its conditional density, $V_{\ref{a:strict-monotonicity}}$, and $U$, in the proof of lemma \ref{l:same-info} in the appendix section of proofs \ref{sec:proofs}. As a consequence of lemma \ref{l:same-info}, conditioning on the feasible control quantity $V_{\ref{a:strict-monotonicity}}$ jointly with $U$ is as good as conditioning on $V = F_{A | Z, U}(A | Z, U)$ (which corresponds to a unique $\eta$) with $U$. Hence, if $U$ were observed, conditioning on $V_{\ref{a:strict-monotonicity}}$ would still identify the causal effect $J$. The conditional exogeneity of $Z$ renders $V_{\ref{a:strict-monotonicity}}$ a valid control function conditional on $U$. We formulate this idea in lemma \ref{t:still-valid}, which formally states that conditioning the IPW, REG and DR estimators on $V_{\ref{a:strict-monotonicity}}$ and $U$ instead of $V$ and $U$ leads to a similarly unbiased estimator. Let the IPW, REG and DR estimators be defined as in lemma \ref{l:simple-cf}.

\begin{theorem} \label{t:still-valid}
Under assumptions \ref{a:icc}, \ref{a:strict-monotonicity} and \ref{a:icc-first-add},
\begin{align*}
\E{\phi_{IPW}(Y, A, V_{\ref{a:strict-monotonicity}}, U; f_{A | V_{\ref{a:strict-monotonicity}}, U})} &= \E{\phi_{IPW}(Y, A, V, U; f_{A | V, U})}, \\
\E{\phi_{REG}(V_{\ref{a:strict-monotonicity}}, U; k_{0, V_{\ref{a:strict-monotonicity}}})} &= \E{\phi_{REG}(V, U; k_{0, v})},  \\
\E{\phi_{DR}(Y, A, V_{\ref{a:strict-monotonicity}}, U; f_{A | V_{\ref{a:strict-monotonicity}}, U}, k_{0, V_{\ref{a:strict-monotonicity}}})} &= \E{\phi_{DR}(Y, A, V, U; f_{A | V, U}, k_{0, v})}.
\end{align*}
\end{theorem}

Theorem \ref{t:still-valid} is one of our main results, but not immediately useful. Neither are $U$ ever observed, nor have we shown how to identify $V_{\ref{a:strict-monotonicity}}$ from observable data. We address the identifiability of control quantity $V_{\ref{a:strict-monotonicity}}$ in section \ref{sssec:control-bridge-id} with bridge functions. 

\subsubsection{Identification of the Control Quantity} \label{sssec:control-bridge-id}
We assume the existence of a control bridge function $\tau_{A,0}$ in assumption \ref{a:control-bridge}. For its definition, we use variables $W_0$ and $W_1$, where $W_0 \in \mathcal{W}_0$ and $W_1 \in \mathcal{W}_1$. $W_0$ and $W_1$ need to have sufficient independent variation conditional on $U$, but they must not be fully conditionally independent. That conditionally independent part of variation in $W_0$ and $W_1$ must be sufficiently relevant for $U$. $W_0$ and $W_1$ may be elements of the outcome-inducing proxy $W \in \mathcal{W}$, or not.
\begin{assumption} \label{a:control-bridge}
Let $(W_0, W_1) \in \mathcal{W}_0 \times \mathcal{W}_1$. There exists a bridge function $\tau_{A, 0} \in L_2(Z, W_1)$ for all $A \in \mathcal{A}$ such that
\begin{align}
\E{ \tau_{A, 0}(Z, W_1)   | Z, W_0, U } &=  F(A | Z, U). \label{eq:control-bridge-1}
\end{align}
and $\kappa_{0} \in L_2(Z, W_0)$ such that
\begin{align}
\E{ \kappa_{0}(Z, W_0) | Z, W_1, U} &= \frac{f(U)}{f(U | Z)}. \label{eq:control-bridge-2}
\end{align}
\end{assumption}
Assumption \ref{a:control-bridge} is a requirement on the amount of independent variation in $W_0$ and $W_1$ conditional on $U$. $W_0$ and $W_1$ must not be fully independent conditional on $U$. However, some independent variation must be contained in $W_0$ and $W_1$, and this independent variation must be sufficiently rich with respect to $U$.

Let $\mathbb{T}_0$ and $\mathbb{K}_0$ 
be the nonempty (assumption \ref{a:control-bridge}) sets of valid control bridge functions defined by
\begin{align}
\mathbb{T}_0 &= \left\{ \tau_A \in L_2(Z, W_1): \E{ \tau_A(Z, W_1)   | Z, U, W_0 } = F(A | Z, U) \right\} \neq \emptyset, \\ 
\mathbb{K}_0 &= \left\{ \kappa \in L_2(Z, W_0): \E{ \kappa(Z, W_0) | Z, U, W_1 } = \frac{f(U)}{f(U| Z)} \right\} \neq \emptyset.
\end{align}
So far, the existence of such bridge functions does not immediately help with identification, because their defining conditional moments involve the unobserved $U$. Using the same proximal learning logic as in section \ref{sec:id-outcome}, lemma \ref{l:control-obs} states that all valid bridge functions above also satisfy a different set of observed conditional moments.
\begin{lemma} \label{l:control-obs}
Under assumption 
\ref{a:icc-first-add}, 
any $\tau_{A,0} \in \mathbb{T}_0$ and $\kappa_{0} \in \mathbb{K}_0$ 
satisfy that
\begin{align}
\E{ \tau_{A,0}(Z, W_1) | Z, W_0}  &= F(A | Z, W_0). \label{eq:control-obs-mom-1} \\
\E{\kappa_{0}(Z, W_0) | Z, W_1} &= \frac{f(W_1)}{f(W_1 | Z)}. \label{eq:control-obs-mom-2}
\end{align}
\end{lemma}
The above result suggests that we may learn the bridge function $\tau_{A,0}$ and $\kappa_0$ 
from the set of observed bridge functions $\mathbb{T}_0^{\text{obs}}$ and $\mathbb{K}_0^{\text{obs}}$, defined below.
\begin{align}
\mathbb{T}_0^{\text{obs}} &= \left\{ \tau_A \in L_2(Z, W_1): \E{ \tau_A(Z, W_1)   | Z, W_0 } = F(A | Z, W_0) \right\} \neq \emptyset, \\
\mathbb{K}_0^{\text{obs}} &= \left\{ \kappa \in L_2(Z, W_0): \E{ \kappa(Z, W_0) | Z, W_1 } = \frac{f(W_1)}{f(W_1 | Z)} \right\} \neq \emptyset.
\end{align}
As no assumptions are imposed on the uniqueness of the observed bridge functions, it remains uncertain how to identify the control quantity $V_{\ref{a:strict-monotonicity}}$ from them. In this regard, lemma \ref{l:control-obs-id} states that any $\tau_A \in \mathbb{T}_0^{\text{obs}}$ can identify $V_{\ref{a:strict-monotonicity}}$.
\begin{lemma} \label{l:control-obs-id}
Suppose assumptions 
\ref{a:icc-first-add} and \ref{a:control-bridge} hold. 
It follows that for any $\tau_A \in L_2(Z, W_1)$ and $\kappa_0 \in \mathbb{K}_0^{\text{obs}}$,
\begin{align*}
\int_{\mathcal{W}_1} \tau_A(Z, W_1) \dif F(W_1) - V_{\ref{a:strict-monotonicity}} &= \E{ \kappa_0(Z, W_0) \E{ \tau_A(Z, W_1) - F(A | Z, W_0) | Z, W_0} | Z }.
\end{align*}
Hence, for any $\tau_A \in \mathbb{T}_0^{\text{obs}}$ as long as $\mathbb{K}_0^{\text{obs}} \neq \emptyset$,
\begin{align*}
V_{\ref{a:strict-monotonicity}} &= \int_{\mathcal{W}_1} \tau_A(Z, W_1) \dif F(W_1).
\end{align*}
\end{lemma}
The above result means that a control function exists and can be estimated from $Z$, $W_0$ and $W_1$ under richness and (some, not full) independence requirements for $W_0$ and $W_1$. Importantly, once the control function is identified, the causal effect $J$ can be estimated with action and outcome bridge functions similarly to standard proximal learning \citep{cui2020}.

\subsection{Action and Outcome Bridge Functions} \label{ssec:bridge-functions}
In subsections \ref{ssec:cfccnifs} and \ref{ssec:control-bridge}, we established that valid control functions exist for the endogenous variation $\eta$ in treatment $A$ when the unobserved common confounders $U$ are absent from the first stage (assumption \ref{a:strict-monotonicity-no-cf-fs}), and when $U$ enter the first-stage in a complexity-restricted way (assumptions \ref{a:strict-monotonicity}, \ref{a:icc-first-add}, \ref{a:control-bridge}). Identification of causal effect $J$ with unobserved common confounders $U$ was not yet explained. In this section, we use ideas from standard proximal learning \citep{cui2020} to address the remaining confounding problem stemming from unobserved $U$ as a source of variation in instruments $Z$. 
Properties of the action and outcome bridge function properties are explained in section \ref{sssec:bridge-functions-properties}, their identification is discussed in section \ref{sssec:bridge-functions-id}.

\subsubsection{Properties of the Action and Outcome Bridge Functions} \label{sssec:bridge-functions-properties}
First, we use assumption \ref{a:bridge-functions} to (re-)specify the exclusion and relevance for $U$ of outcome-inducing proxies $W$. The outcome bridge \ref{eq:outcome-bridge} for the regression $k_{0, \tilde{V}}(A, \tilde{V}, U)$ sufficiently describes the relevance requirement for $W$, and the action bridge \ref{eq:action-bridge} for $Z$ with the generalised propensity score (with contrast function) $\frac{\pi(A)}{f(A | \tilde{V}, U)}$. The existence of the bridge functions is the key identification assumption. For example, the action bridge \ref{eq:action-bridge} existence poses a stronger richness requirement on $Z$ compared to standard proximal learning \citep{kallus2021}. Intuitively, the instruments $Z$ must be sufficiently relevant for $A$ and $U$, where the latter is satisfied by definition when $U$ captures all unobserved variables conditional on which $Z$ would satisfy an exclusion restriction. 
\begin{assumption}[Bridge functions] \label{a:bridge-functions}
$W(a, z) = W \CI (A, Z) \ | \ U$. 
Let $\tilde{V} = V_{\ref{a:strict-monotonicity-no-cf-fs}}$ if \ref{a:strict-monotonicity-no-cf-fs} holds, and $\tilde{V} = V_{\ref{a:strict-monotonicity}}$ if \ref{a:strict-monotonicity} holds.
There exist some function $h_0 \in L_2(A, W)$ and $q_0 \in \pi L \in L_2(A, Z)$ such that 
\begin{align}
\E{h_0(A, \tilde{V}, W) | A, \tilde{V}, U} &= k_{0, \tilde{V}}(A, \tilde{V}, U), \label{eq:outcome-bridge} \\
\E{\pi(A) q_0(A, \tilde{V}, Z) | A, \tilde{V}, U} &= \frac{\pi(A)}{f(A | \tilde{V}, U)} \label{eq:action-bridge}
\end{align}
almost surely.
\end{assumption}
Equation \ref{eq:outcome-bridge} is the outcome bridge function and \ref{eq:action-bridge} the action bridge function. While their existence is assumed, their uniqueness is not. 

Let $\mathbb{H}_0$ and $\mathbb{Q}_0$ be the nonempty (assumption \ref{a:bridge-functions}) sets of valid outcome and action bridge functions defined by
\begin{align}
\mathbb{H}_0 &= \left\{ h \in L_2(A, \tilde{V}, W): \E{Y - h(A, \tilde{V}, W) | A, \tilde{V}, U} = 0 \right\} \neq \emptyset, \\
\mathbb{Q}_0 &= \left\{ q \text{ with } \pi q \in L_2(A, \tilde{V}, Z): \E{\pi(A) (q(A, \tilde{V}, Z) - 1/f(A|\tilde{V},U)) | A, \tilde{V}, U} = 0 \right\} \neq \emptyset.
\end{align}
Any valid bridge function can identify $J$. We adapt lemma 2 from \cite{kallus2021} to our common confounding model for lemma \ref{l:main}. The causal effect $J$ is identified in the common confounding model with any valid bridge function.
\begin{lemma} \label{l:main}
Let $\mathcal{T}: L_2(A, \tilde{V}, W) \rightarrow L_2(\tilde{V}, W)$ be the linear operator defined by $(\mathcal{T}h)(\tilde{v}, w) = \int_{\mathcal{A}} h(a, \tilde{v}, w) \pi(a) d\mu_A(a)$. Suppose \ref{a:common-support-simple} and \ref{a:bridge-functions} hold. 
Suppose either \ref{a:strict-monotonicity-no-cf-fs} ($\tilde{V} = V_{\ref{a:strict-monotonicity-no-cf-fs}}$), or (\ref{a:strict-monotonicity}, \ref{a:icc-first-add}, \ref{a:control-bridge}) ($\tilde{V} = V_{\ref{a:strict-monotonicity}}$) hold.
Then, for any $h_0 \in \mathbb{H}_0$ and $q_0 \in \mathbb{Q}_0$,
\begin{align*}
J = &\E{\tilde{\phi}_{IPW}(Y, A, \tilde{V}, Z; q_0)} = \E{\tilde{\phi}_{REG}(\tilde{V}, W; h_0)} = \E{\tilde{\phi}_{DR}(Y, A, \tilde{V}, W; h_0, q_0)} \\
\text{where} & \ \ \ \ \tilde{\phi}_{IPW}(y, a, \tilde{v}, z; q_0) = y \pi(a) q_0(a,\tilde{v},z) \\
& \ \ \ \ \tilde{\phi}_{REG}(\tilde{v}, w; h_0) = (\mathcal{T} h_0)(\tilde{v}, w) \\
& \ \ \ \ \tilde{\phi}_{DR}(y, a, \tilde{v}, z, w; h_0, q_0) = (y - h_0(a, \tilde{v}, w)) \pi(a) q_0(a, \tilde{v}, z) + (\mathcal{T} h_0)(\tilde{v}, w) 
\end{align*}
\end{lemma}
Notably, all moment functions in lemma \ref{l:main} are functions of observed variables only. By simple replacement of the infeasible regression $k_0(A, V, U)$ and generalised propensity score $f^{-1}(A | V, U)$ with their feasible counterpart bridge functions $h_0(A, \tilde{V}, W)$ and $q_0(A, \tilde{V}, Z)$, we can identify causal effect $J$. Again, this is only useful if we can identify the bridge functions.

\subsubsection{Identification of the Action and Outcome Bridge Functions} \label{sssec:bridge-functions-id}

The existence of bridge functions as defined in assumption \ref{a:bridge-functions} is not immediately useful, because these bridge functions are defined conditional on the unobserved common confounders $U$. Fortunately, these bridge functions also satisfy some conditional moment equations involving only observable data. This result is stated in lemma \ref{l:obs} below, where the bridge functions defined in assumption \ref{a:bridge-functions} are solutions to feasible conditional moment equations.
\begin{lemma} \label{l:obs}
Suppose \ref{a:common-support-simple} and \ref{a:bridge-functions} hold ($\mathbb{H}_0 \neq \emptyset$ and $\mathbb{Q}_0 \neq \emptyset$). 
Suppose either \ref{a:strict-monotonicity-no-cf-fs} ($\tilde{V} = V_{\ref{a:strict-monotonicity-no-cf-fs}}$), or (\ref{a:strict-monotonicity}, \ref{a:icc-first-add}, \ref{a:control-bridge}) ($\tilde{V} = V_{\ref{a:strict-monotonicity}}$) hold.
Then, any $h_0 \in \mathbb{H}_0$ and $q_0 \in \mathbb{Q}_0$ satisfy that
\begin{align}
\E{\left. Y - h_0(A, \tilde{V}, W) \right| A, \tilde{V}, Z} &= 0 \label{eq:obs-mom-1} \\
\E{\left. \pi(A) \left(q_0(A, \tilde{V}, Z) - \frac{1}{f(A|\tilde{V},W)}\right) \right| A, \tilde{V}, W} &= 0 \label{eq:obs-mom-2}
\end{align}
\end{lemma}
With these conditional moments of observable data we may identify a different set of observable bridge functions.
\begin{align}
\mathbb{H}^\text{obs}_0 &= \{ h \in L_2(A, \tilde{V}, W): \E{\left. Y - h(A, \tilde{V}, W) \right| A, \tilde{V}, Z} = 0 \} \neq \emptyset \\
\mathbb{Q}^\text{obs}_0 &= \{ q \text{ with } \pi q \in L_2(A, \tilde{V}, Z): \E{\left. \pi(A) (q(A, \tilde{V}, Z) - 1/f(A|\tilde{V},W)) \right| A, \tilde{V}, W} = 0 \} \neq \emptyset
\end{align}
From lemma \ref{l:obs} we know that all proper bridge functions according to assumption \ref{a:bridge-functions}, which are in the sets $\mathbb{H}_0$ and $\mathbb{Q}_0$, are also in $\mathbb{H}^\text{obs}_0$ and $\mathbb{Q}^\text{obs}_0$: $\mathbb{H}_0 \subseteq \mathbb{H}^\text{obs}_0$ and $\mathbb{Q}_0 \subseteq \mathbb{Q}^\text{obs}_0$. Assumption \ref{a:bridge-functions} for the existence of bridge functions ensures $\mathbb{H}^\text{obs}_0 \neq \emptyset$ and $\mathbb{Q}^\text{obs}_0 \neq \emptyset$. We follow the logic in \cite{kallus2021} for the common confounding model to show that any $h_0 \in \mathbb{H}^\text{obs}_0$ and $q_0 \in \mathbb{Q}^\text{obs}_0$, which satisfy the observed conditional moments \ref{eq:obs-mom-1} and \ref{eq:obs-mom-2}, also identify the causal effect $J$.
\begin{lemma} \label{l:5}
Suppose \ref{a:common-support-simple} holds, and $W \CI (A, Z) \ | \ U$.
Suppose either \ref{a:strict-monotonicity-no-cf-fs} ($\tilde{V} = V_{\ref{a:strict-monotonicity-no-cf-fs}}$), or (\ref{a:strict-monotonicity}, \ref{a:icc-first-add}, \ref{a:control-bridge}) ($\tilde{V} = V_{\ref{a:strict-monotonicity}}$) hold.
Take any $h_0 \in \mathbb{H}^\text{obs}_0$, assuming $\mathbb{H}^\text{obs}_0 \neq 0$ and $\mathbb{Q}_0 \neq \emptyset$. Then, 
\begin{align}
\E{\tilde{\phi}_{IPW}(Y, A, \tilde{V}, Z; q) - J} &= \E{h_0(A, \tilde{V}, W) \E{\pi(A) (q(A, \tilde{V}, Z) - 1/f(A | \tilde{V}, W)) | A, \tilde{V}, W}}. \label{eq:l5-ipw}
\end{align}
Take any $q_0 \in \mathbb{Q}^\text{obs}_0$, assuming $\mathbb{Q}^\text{obs}_0 \neq 0$ and $\mathbb{H}_0 \neq \emptyset$. Then, 
\begin{align}
\E{\tilde{\phi}_{REG}(\tilde{V}, W; h) - J} &= \E{\pi(A) q_0(A, \tilde{V}, Z) \E{h(A, \tilde{V}, W) - Y | A, \tilde{V}, Z}}. \label{eq:l5-reg}
\end{align}
\end{lemma}

From lemma \ref{l:5}, it follows that for any $h_0 \in \mathbb{H}^\text{obs}_0$ and $q_0 \in \mathbb{Q}^\text{obs}_0$, all three estimators identify causal effect $J$. We state this result formally in theorem \ref{t:4}.
\begin{theorem} \label{t:4}
Suppose \ref{a:common-support-simple} holds, and $W \CI (A, Z) \ | \ U$.
Suppose either \ref{a:strict-monotonicity-no-cf-fs} ($\tilde{V} = V_{\ref{a:strict-monotonicity-no-cf-fs}}$), or (\ref{a:strict-monotonicity}, \ref{a:icc-first-add}, \ref{a:control-bridge}) ($\tilde{V} = V_{\ref{a:strict-monotonicity}}$) hold.

Suppose $\mathbb{Q}_0 \neq \emptyset$ and $\mathbb{H}_0^{\text{obs}} \neq \emptyset$. Then, for any $q_0 \in \mathbb{Q}_0^{\text{obs}}$
\begin{align}
\E{\tilde{\phi}_{IPW}(Y, A, \tilde{V}, Z; q_0)} &= J. \label{eq:l4-ipw}
\end{align}

Suppose $\mathbb{H}_0 \neq \emptyset$ 
and $\mathbb{Q}_0^{\text{obs}} \neq \emptyset$. Then, for any $h_0 \in \mathbb{H}_0^{\text{obs}}$
\begin{align}
\E{\tilde{\phi}_{REG}(\tilde{V}, W; h_0)} &= J. \label{eq:l4-reg}
\end{align}

Suppose either $\{ \mathbb{Q}_0 \neq \emptyset$ and $\mathbb{H}_0^{\text{obs}} \neq \emptyset \}$ or $\{ \mathbb{H}_0 \neq \emptyset$ and $\mathbb{Q}_0^{\text{obs}} \neq \emptyset \}$.
Then, for any $h_0 \in \mathbb{H}_0^\text{obs}$ and $q_0 \in \mathbb{Q}_0^\text{obs}$,
\begin{align}
\E{\tilde{\phi}_{DR}(Y, A, \tilde{V}, Z, W; h_0, q_0)} &= J. \label{eq:l4-dr}
\end{align}
\end{theorem}
This key result in theorem \ref{t:4} establishes the identifiability of causal effect $J$ from only observable data in the common confounding model with some first stage reduced form monotonicity assumption.

As a side note, only a relaxed version of assumption \ref{a:bridge-functions} is needed, because e.g. unbiasedness of the regression estimator $\tilde{\phi}_{REG}$ only requires $\mathbb{Q}^{\text{obs}}_0 \neq \emptyset$ instead of $\mathbb{Q}_0 \neq \emptyset$. With the doubly robust estimator, the relaxation is even more explicit in theorem \ref{t:4}. Despite these technically feasible relaxations, completeness assumptions on the structural form, i.e. the completeness of $Z$ and $W$ with respect to $U$, likely are easier to comprehend and already imply $\mathbb{H}_0 \neq \emptyset$ and $\mathbb{Q}_0 \neq \emptyset$. 
Corollary \ref{cor:l4-stronger} states a simplified version of theorem \ref{t:4} with these slightly stronger assumptions ($\mathbb{H}_0 \neq \emptyset$ and $\mathbb{Q}_0 \neq \emptyset$).

\begin{corollary} \label{cor:l4-stronger}
Suppose \ref{a:common-support-simple} and \ref{a:bridge-functions} hold ($\mathbb{H}_0 \neq \emptyset$ and $\mathbb{Q}_0 \neq \emptyset$). 
Suppose either \ref{a:strict-monotonicity-no-cf-fs} ($\tilde{V} = V_{\ref{a:strict-monotonicity-no-cf-fs}}$), or (\ref{a:strict-monotonicity}, \ref{a:icc-first-add}, \ref{a:control-bridge}) ($\tilde{V} = V_{\ref{a:strict-monotonicity}}$) hold. Then, for any $h_0 \in \mathbb{H}_0^\text{obs}$ and $q_0 \in \mathbb{Q}_0^\text{obs}$,
\begin{align}
J = &\E{\tilde{\phi}_{IPW}(Y, A, \tilde{V}, Z; q_0)} = \E{\tilde{\phi}_{REG}(Y, A, \tilde{V}, W; h_0)} = \E{\tilde{\phi}_{DR}(Y, A, \tilde{V}, Z, W; h_0, q_0)}.
\end{align}
\end{corollary}

\section{Numerical Examples} \label{sec:num-examples}
We provide examples of bridge functions in linear, discrete and nonparametric models in this subsection.

\subsection{Linear model}  \label{ssec:example-linear}
Let $Y_i \in \mathbb{R}$, $A_i \in \mathbb{R}$, $U_i \in \mathbb{R}^{d_U}$, $Z_i \in \mathbb{R}^{d_Z}$ and $W_i \in \mathbb{R}^{d_W}$, for observation $i \in \{ 1, 2, \hdots, n \}$. The parameter of interest is $\beta$, the linear effect of treatment $A$ on outcome $Y$.

\begin{align*}
Y_i &= A_i {\beta} + U_i \underset{d_U \times 1}{\gamma_Y} + W_i \underset{W \times 1}{\zeta} + \varepsilon_{Y, i} & \E{\varepsilon_{Y, i} | Z_i} &= 0,   \\
A_i &= Z_i \underset{d_Z \times d_A}{\pi} + U_i \underset{d_U \times d_A}{\gamma_A} + \varepsilon_{A, i} & \E{\varepsilon_{A, i} | Z_i} &= 0, \ \pi \neq 0, \ d_Z \geq (d_A + d_U) \\
W_i &= U_i \underset{d_U \times d_W}{\gamma_W} + \varepsilon_{W, i} & \E{\varepsilon_{W, i} | Z_i} &= 0
\end{align*}
Any bias in the IV estimator stems from changes in $\E{U_i | Z_i}$ as $Z_i$ changes. Suppose the parameter matrix $\gamma_W$ has full rank. Then, if $d_W \geq d_U$, holding $\E{W_i | Z_i}$ constant is equivalent to holding $\E{U_i | Z_i}$ constant. Consequently, conditioning on $\E{W_i | Z_i}$, which itself has at most column rank $d_U \leq d_W$, eliminates all confounding bias from the common confounders $U$. 
We could write the parameter of interest, $\beta$, in the ratio form
\begin{align*}
\beta &= \frac{ \Cov{ Y_i, Z_i \pi \left| \E{W_i | Z_i} \right. } }{ \Cov{A_i, Z_i \pi \left| \E{W_i | Z_i} \right.} }.
\end{align*}
From the above formula follows a consistent estimator $\hat{\beta}_{ICC}$, which we can write in matrix notation. Let $(Y, A, Z, W, U)$ be the $n$-rowed vectors and matrices of $n$ stacked sample observations. We let $P_Z \coloneqq Z ( Z^\intercal Z)^{-1} Z^\intercal$. Then, the estimator of the $n \times d_W$ matrix $\E{W | Z}$ is $\hat{W} \coloneqq P_Z W$. For simplicity, suppose that $d_W = d_U$, which implies that $\E{W | Z}$ has full column rank. If $d_W > d_U$, we would use exactly $d_U$ linearly independent columns of the $\E{W | Z}$ matrix. Let
\begin{align*}
P_{\hat{W}} &\coloneqq \hat{W} \left( \hat{W}^\intercal \hat{W} \right)^{-1} \hat{W}^\intercal \\
 &= P_Z W \left( W^\intercal P_Z W \right)^{-1} W^\intercal P_Z, \\
M_{\hat{W}} &\coloneqq I_n - P_{\hat{W}} = I_n - P_Z W \left( W^\intercal P_Z W \right)^{-1} W^\intercal P_Z
\end{align*}
The reduced form of the outcome model and the ICC estimator can be written as
\begin{align*}
Y &= A \beta + \E{W | Z} \delta + \varepsilon_{ICC}, \ \ \ \ \delta \coloneqq \gamma_W^\intercal \left( \gamma_W \gamma_W^\intercal \right)^{-1} \left( \gamma_Y + \gamma_W \zeta \right), \\ 
\varepsilon_{ICC} &\coloneqq \varepsilon_Y + \varepsilon_W \zeta + \left( U - \E{U | Z} \right)  \left( \gamma_Y + \gamma_W \zeta \right), \\ 
\hat{\beta}_{ICC} &= \left( A^\intercal M_{\hat{W}} P_Z A \right)^{-1} \left( A^\intercal M_{\hat{W}} P_Z Y \right).
\end{align*}
The asymptotics of the estimator then follow from a few simple steps of algebra using a central limit theorem under standard assumptions.
\begin{align*}
\sqrt{n} \left( \hat{\beta}_{ICC} - \beta \right)
&= \left( \frac{1}{n} A^\intercal  P_Z M_{\hat{W}} A \right)^{-1} \left( \frac{1}{\sqrt{n}} A^\intercal M_{\hat{W}} P_Z \varepsilon_{ICC} \right) \\
&\overset{d}{\rightarrow} N \left( 0, \left( A^\intercal  P_Z M_{\hat{W}} A \right)^{-1} \left( A^\intercal  M_{\hat{W}} P_Z \E{  \varepsilon_{ICC} \varepsilon_{ICC}^\intercal } P_Z M_{\hat{W}} A \right) \left( A^\intercal  P_Z M_{\hat{W}} A \right)^{-1} \right).
\end{align*}
Another interpretation is to separate the estimator into stages, like 2SLS. Here, we get the predicted values $\hat{A} \coloneqq P_Z A$ and $\hat{W} = P_Z W$ in the first stage. Then we regress $Y$ on the predicted values $\hat{A}$ and $\hat{W}$. The OLS estimates in the second stage are the point estimates of the ICC estimator $\hat{\beta}_{ICC}$ and the auxiliary OLS slope parameter $\hat{\delta}_{ICC}$. The correct standard errors are easy to calculate from the above asymptotic distribution once we realise that $\hat{\varepsilon}_{ICC} = Y - A \hat{\beta}_{ICC} - P_Z W \hat{\delta}_{ICC}$. Estimating $\pi$ via OLS in the first stage conditional on $\hat{W}$, we could have set up a numerically equivalent control function estimator using $A - Z \hat{\pi}$ and $\hat{W}$.

Importantly, the variation in instruments $Z$ must still be relevant for treatment $A$ after conditioning on the predicted values of $W$ given $Z$, $P_Z W$. This requirement is generally asymptotically satisfied when $d_Z \geq (d_A + d_U)$ with individually relevant instruments for a $d_A$-dimensional treatment $A$ (except for some numerical special cases where $Z$ correlates with $A$ identically as with $U$). Despite this asymptotic guarantee, sample variation implies that the estimator $P_Z W$ always has the largest possible column rank $\min \{d_W, d_Z\}$, even when the true column rank of its estimand $\E{W | Z}$ is $d_U < \min \{d_W, d_Z \}$. Thus, whenever $d_W \geq d_Z$, the rank of the estimator $P_Z W$ would be $d_Z$, and therefore must be reduced to some rank $r$, which satisfies $d_U \geq r < d_Z$. As $d_U$. This necessary in-sample rank reduction is justified by the true rank $d_U$ of estimand $\E{W | Z}$, to which the estimator $P_Z W$ converges.

\subsection{Discrete model}  \label{ssec:example-discrete}
\subsubsection{Discrete model with outcome model restrictions}  \label{sssec:discrete-outcome}
In the case of outcome model restrictions, recall that 
\begin{align*}
Y = Y(A) &= {k_0(A, U)} + \varepsilon, & \E{\varepsilon | Z, U} &= 0.
\end{align*}
Let $A, Z, W, U$ be discrete variables, which take values $a_{j_a}$, $z_{j_z}$, $w_{j_w}$, $u_{j_U}$ for $j_a = 1, \hdots, d_A$, $j_z = 1, \hdots, d_Z$, $j_w = 1, \hdots, d_W$ and $j_U = 1, \hdots, d_U$.
To define the control bridge function \ref{eq:npiv-bridge-outcome}, we first define some useful matrices of conditional probabilities:
\begin{align*}
\underset{1 \times d_U}{K_{0}(A=a_{j_A}, \pmb{U})} &\text{ has $(j_{U})$th element } {k_0(A=a_{j_{A}}, U=u_{j_{U}})}, \\
\underset{d_W \times d_U}{P \left( \pmb{W} | \pmb{U} \right)} &\text{ has $(j_{W}, j_{U})$th element } \Prp{W=w_{j_{W}} \left| U=u_{j_{U}} \right.} \\ 
\underset{(d_A \times d_W) \times d_Z}{P \left( \pmb{A}, \pmb{W} | \pmb{Z} \right)} &\text{ has $((j_{A} \times j_{W}), j_{U})$th element } \Prp{A=a_{j_{A}}, W=w_{j_{W}} \left| U=u_{j_{U}} \right.}, \\
\underset{1 \times d_Z}{\E{Y | \pmb{Z}}} &\text{ has $j_{Z}$th element } \E{Y | Z=z_{j_{Z}}}.
\end{align*}

The control bridge function \ref{eq:npiv-bridge-outcome} here corresponds to the linear system
\begin{align*}
\underset{1 \times d_W}{H_0\left(A=a_{j_A}, \pmb{W}\right)} \underset{d_W \times d_U}{P\left(\pmb{W} | \pmb{U} \right)}  &= \underset{1 \times d_U}{K_{0}\left(A=a_{j_A}, \pmb{U}\right)}, 
\end{align*}
Now if $P\left(\pmb{W} | \pmb{U} \right)$ has full column rank, which requires $d_W \geq d_U$, the linear equations system above has a solution. Unless $d_W = d_U$, there will be more than one solution and the bridge function $H_0\left(A=a_{j_A}, \pmb{W}\right)$ will be non-unique. As the common confounder $U$ is not observed, we can at best identify some other bridge function which satisfies
\begin{align*}
\underset{1 \times (d_A \times d_W)}{ H_0^{\text{obs}}(\pmb{A}, \pmb{W}) } \underset{(d_A \times d_W) \times d_Z}{P \left( \pmb{A}, \pmb{W} | \pmb{Z} \right)} &= \underset{1 \times d_Z}{\E{Y | \pmb{Z}}}.
\end{align*}
The existence of this bridge function is guaranteed if $d_W \geq d_U$. Our results imply that any solution $H_0^{\text{obs}}(\pmb{A}, \pmb{W})$ results in a unique 
$$H_0^{\text{obs}}(A=a_{j_A}, \pmb{W}) P(\pmb{W}) = K_{0}\left(A=a_{j_A}, \pmb{U}\right) P \left( \pmb{U} \right),$$ 
which can be interpreted as an average structural function where the confounders $U$ are integrated out without dependence on treatment $A$. Required for this result is the sufficient relevance of instruments $Z$ for treatment $A$ conditional on common confounders $U$, i.e. $d_Z \geq (d_U \times d_A)$.

\subsubsection{Discrete model with first stage restrictions} \label{sssec:discrete-first}
Let $Z, W, U$ be discrete variables, which take values $z_{j_z}$, $w_{j_w}$, $u_{j_U}$ for $j_z = 1, \hdots, d_Z$, $j_w = 1, \hdots, d_W$ and $j_U = 1, \hdots, d_U$. 
The matrices of conditional probabilities, vectors of conditional expectations and densities are
\begin{align*}
\underset{d_{W_k} \times d_U}{P(\pmb{W}_k | \pmb{U}, w_i)} &\text{ has $(j_{W_k}, j_U)$th element } \Prp{W=w_{j_{W_k}} | U=u_{j_U}, W_0=w_0} \text{ for } k \neq i \in \{0, 1\}, \\
\underset{d_W \times d_U}{P(\pmb{W} | \pmb{U}, a, \tilde{v})} &\text{ has $(j_W, j_U)$th element } \Prp{W=w_{j_W} | U=u_{j_U}, A=a, \tilde{V}=v}, \\
\underset{d_Z \times d_U}{P(\pmb{Z} | \pmb{U}, a, \tilde{v})} &\text{ has $(j_Z, j_U)$th element } \Prp{Z=w_{j_Z} | U=u_{j_U}, A=a, \tilde{V}=v}, \\
\underset{1 \times d_U}{\E{Y | \pmb{U}, a, \tilde{v}}} &\text{ has $j_U$th element } \E{Y | U=u_{j_U}, A=a, \tilde{V}=\tilde{v}}, \\
\underset{1 \times d_U}{F^*(a | \pmb{U}, z)} &\text{ is a vector with $j_U$th entry } f(a | u_{j_U}, z), \\
\underset{d_U \times 1}{F^*(a | \pmb{U}, \tilde{v})} &\text{ is a vector with $j_U$th entry } f(a | u_{j_U}, \tilde{v}), \\
\underset{1 \times d_U}{F(\pmb{U}, z)} &\text{ is a vector with $j_U$th entry } \frac{f(u_{j_U})}{f(u_{j_U} | z)}, \\
\text{and } \underset{d_U \times 1}{\pmb{e}} &\text{ is a vector of ones.}
\end{align*}
First, the control bridge functions \ref{eq:control-bridge-1} and \ref{eq:control-bridge-2} correspond to the linear system
\begin{align*}
\tau_{A,0}(z, \pmb{W}_1) \Prp{ \pmb{W}_1 | \pmb{U}, w_0} &= F^*(a | \pmb{U}, z) \text{ for all } w_0 \in \mathcal{W}_0, \\
\kappa_0(z, \pmb{W}_0) \Prp{ \pmb{W}_0 | \pmb{U}, w_1} &= F(\pmb{U}, z) \text{ for all } w_1 \in \mathcal{W}_1.
\end{align*}
A sufficient condition for the solutions to exist is that $\Prp{ \pmb{W}_1 | \pmb{U}, w_0}$ and $\Prp{ \pmb{W}_0 | \pmb{U}, w_1}$ have full rank for all $(w_0, w_1) \in \mathcal{W}$. This requires $W_0$ and $W_1$ to contain a minimum amount of conditionally independent categories. Let $\mathcal{W}^{\CI}$ be the union of all largest-possible non-overlapping subsets $\mathcal{W}^{\CI}_i \in \mathcal{W}$ of values for which $(W_0, W_1) \in \mathcal{W}^{\CI}_i$ are conditionally independent. $i \in \{1, 2, \hdots, \mathbf{card}(\mathcal{W}^{\CI})\}$ indexes each subset of $\mathcal{W}^{\CI}_i \in \mathcal{W}$.
A sufficient condition for the existence of the bridge functions is that for each subset $i$, there are at least as many categories $d_{W_k, i}$ for both $k \in \{ 0, 1 \}$ as categories of the common confounder $d_U$: $d_{W_k, i} \geq d_U$ for all subsets $i$ and both $k \in \{ 0, 1 \}$. If $W_1 \CI W_0$ everywhere, this condition is satisfied if $d_{W_k} > d_U$ for both $k \in \{ 0, 1 \}$. However, such a full conditional independence assumption often is too strong in observational data. 


The bridge function equations for outcome \ref{eq:outcome-bridge} and action \ref{eq:action-bridge} correspond to the linear system
\begin{align*}
h_0^\intercal(\pmb{W}, a, \tilde{v}) P(\pmb{W} | \pmb{U}, a, \tilde{v}) &= \E{Y | \pmb{U}, a, \tilde{v}}, \\
q_0^\intercal(\pmb{Z}, a, \tilde{v}) P(\pmb{Z} | \pmb{U}, a, \tilde{v}) F^*(a | \pmb{U}, \tilde{v}) &= \pmb{e}^\intercal,
\end{align*}
in this model with discrete $Z, W, U$. Suppose that $P(\pmb{W} | \pmb{U}, a, \tilde{v})$ and $P(\pmb{Z} | \pmb{U}, a, \tilde{v})$ are full rank, which requires $d_W \geq d_U$ and $d_Z \geq (d_A \times d_U)$. Also suppose that $f(a | u, \tilde{v}) \neq 0$ for any $u \in \mathcal{U}$ and $\tilde{v} \in (0, 1)$. Then, the linear system has solutions. Bridge functions exist. The solutions must not be unique unless $d_W=d_U$ and $d_Z = (d_A \times d_U)$. 

Nonsingularity of $P(\pmb{Z} | \pmb{U}, a, \tilde{v})$ enforces a strong richness requirement for instruments $Z$. A necessary condition is that $\Prp{z_{j_Z} | a, \tilde{v}} \in (0, 1)$ for at least $d_U$ different $z_{j_Z} \in \mathcal{Z}$ for each $(a, \tilde{v}) \in \mathcal{A} \times (0, 1)$. 
In simple terms, there must be sufficient relevant information in $Z$ for $A$ conditional on $U$. 

\subsection{Nonparametric model} \label{ssec:example-nonparametric}
\subsubsection{Nonparametric model with outcome model restrictions} \label{sssec:nonparametric-first}
Once we assume some outcome model restrictions (\ref{a:outcome-model}), identification of the causal effect $J$ relies on richness requirements for $W$ with respect to $U$ (\ref{a:icc}.\ref{a:nc-relevance}) and for $Z$ with respect to $U$ and $A$ (\ref{a:icc}.\ref{a:iv-complete}). The existence of bridge function \ref{eq:npiv-bridge-outcome} ensures that outcome-inducing proxies $W$ are sufficiently rich with respect to $U$ to obtain unbiased estimates of the causal effect $J$. As this bridge function is a solution to a Fredholm integral equation of the first kind, its existence is given by Picard's theorem \citep{polyanin2008}. Sufficient for the existence of this bridge function is a completeness condition: For any $g \in L_2(A, U)$,
\begin{align*}
\E{g(A, U) | A, W} &= 0 \text{ only when } g(A, U) = 0.
\end{align*}
This relevance requirement for the outcome-inducing proxies $W$ is identifical to standard proximal learning \citep{kallus2021}. Different is the completeness condition on $Z$, which we included as one of our main assumptions \ref{a:icc}.\ref{eq:iv-complete}. For any $g \in L_2(A, U)$,
\begin{align*}
\E{g(A, U) | Z} &= 0 \text{ only when } g(A, U) = 0.
\end{align*}
As $U$ captures all endogeneity in the instruments $Z$, we require all exogenous variation in the instruments $Z$ to be complete with respect to treatment $A$. In ICC, both conditions must hold simultaneously. The instruments we can consider with these different idenifying assumptions are significantly different from to traditional instruments. Our examples in section \ref{sec:examples} will provide further intuition.

\subsubsection{Nonparametric model with first stage restrictions} \label{sssec:nonparametric-first}
Identification of the conditional moment equations for control (\ref{eq:control-bridge-1}, \ref{eq:control-bridge-2}), outcome (\ref{eq:outcome-bridge}) and action (\ref{eq:action-bridge}) bridge functions does not require parametric assumptions other than the strict monotonicity assumption for the first stage reduced form (\ref{a:strict-monotonicity}).
In nonparametric models, richness requirements for $Z$, $W$, and possibly ($W_0$, $W_1$) with respect to $U$ can be completeness conditions. The bridge functions are Fredholm integral equations of the first kind. The existence of solutions to these inverse learning problems are given by Picard's theorem \citep{polyanin2008}. 
Solutions to the control bridge functions (\ref{eq:control-bridge-1} and \ref{eq:control-bridge-2}) exist if the below completeness conditions are satisfied. 
For any $g \in L_2(U, Z)$,
\begin{align*}
\E{ g(U,Z) | W_{1i}, Z } &= 0 \text{ for } W_{1i} \in \mathcal{W}_{1i} \text{ and any } i \in \{1, 2, \hdots, \mathbf{card}(\mathcal{W}^{\CI})\} \text{ only when } g(U,Z) = 0, \\
\E{ g(U,Z) | W_{0i}, Z } &= 0 \text{ for } W_{0i} \in \mathcal{W}_{0i} \text{ and any } i \in \{1, 2, \hdots, \mathbf{card}(\mathcal{W}^{\CI})\}  \text{ only when } g(U,Z) = 0.
\end{align*}
These completeness conditions require that there is rich enough variation in $W_0$ (and $W_1$) compared to $U$ after conditioning on $W_1$ (or $W_0$).

Similarly, outcome (\ref{eq:outcome-bridge}) and action (\ref{eq:action-bridge}) bridge functions exist if the below completeness conditions are satisfied. For any $g_{U,A,\tilde{V}} \in L_2(U, A, \tilde{V})$,
\begin{align*}
\E{g_{U,A,\tilde{V}}(U, A, \tilde{V}) | Z, A, \tilde{V}} &= 0 \text{ only when } g(U, A, \tilde{V}) = 0, \\
\E{g_{U,A,\tilde{V}}(U, A, \tilde{V}) | W, A, \tilde{V}} &= 0 \text{ only when } g(U, A, \tilde{V}) = 0.
\end{align*}
Only when $Z$ and $W$ vary in a sufficiently rich way with $U$ conditional on control function $\tilde{V}$, the above conditions can hold. Intuitively, the first completeness condition means that after using some variation in $Z$ for the construction of $\tilde{V}$, the remaining variation in $Z$ varies with $U$ in a rich manner. In the linear model above we saw this implies that the dimension of the instrument $d_Z$ must be at least as large as the dimension of treatment and unobserved confounder ($d_A + d_U$). 
The completeness assumptions may be stronger than necessary for the existence of bridge functions. In this paper, we follow the approach by \cite{kallus2021} and instead use the weaker, yet sufficient assumption of existence of bridge functions as much as possible.

\section{ICC in Practice} \label{sec:examples}
First, we present a practical algorithm to construct an ICC model. 
Then, the algorithm is applied to an economic and a medical causal inference problem.
In the economic example, we provide evidence of the benefit of instrumented common confounding in the estimation of returns to education. The medical example attempts to identify the causal effect of a health treatment in observational data.

\subsection{Constructing an ICC Model}  \label{ssec:icc-algo-theory}
In comparison to IV, the ICC assumptions may seem daunting at first. The validity of any identification approach ultimately relies on the theoretical arguments about its untestable assumptions. To facilitate the theoretical discussion about the validity of ICC assumptions, we provide a simple algorithm to check whether an identification problem fits the ICC framework.

\begin{enumerate}
  \item $Z$ exogeneity (\ref{a:icc}.\ref{a:iv-exog}): Define common confounders $U$ such that \begin{align*} Y(a, z) = Y(a) \CI Z \ | \ U \ \forall a \in \mathcal{A}. \end{align*}
  \item $W$ exogeneity (\ref{a:icc}.\ref{a:nc-exog}): Include in the common confounders $U$ any unobserved variables necessary to justify \begin{align*} W(a, z) = W \CI (A, Z) \ | \ U. \end{align*}
  \item $W$ relevance (\ref{a:icc}.\ref{a:nc-relevance}): Check whether $W$ is complete with respect to $U$ given $A$: \begin{align*} \E{g(A,U)|A,W} = 0 \text{ only when } g(A,U)=0 \text{ for any } g \in L_2(A,U). \end{align*}
  \item $Z$ relevance (\ref{a:icc}.\ref{a:iv-complete}): Check whether $Z$ is complete with respect to $(A,U)$: \begin{align*} \E{g(A,U)|Z} = 0 \text{ only when } g(A,U)=0 \text{ for any } g \in L_2(A,U). \end{align*}
\end{enumerate}

The above algorithm requires completeness of $W$ with respect to common confounder $U$ (\ref{a:icc}.\ref{a:nc-relevance}). This requirement can hold even if some variable in $W$ directly causes $(Z,A)$, where such variables in $W$ would be included in $U$, as long as $W$ remains complete for $U$. Vice-versa, some instruments $Z$ might be endogenous with a direct impact on outcome $Y$, where these instruments would be included in $U$. Exogenous variation in the instruments must remain complete with respect to the treatment $A$, and $W$ complete with respect to the endogenous instruments included in $U$. 
Clearly, the completeness assumptions are more likely to be satisfied the richer the information in $Z$ and $W$. Put simply, the more measurements of the unobserved confounders there are, the better we can justify the required completeness conditions. In our subsequent examples, we hope to convince the reader of the relevance of the common confounding assumptions for interesting causal inference problems. Our logic will be based on the above algorithm.

\subsection{Returns to Education} \label{ssec:ex-educ}
The education production function has been a key function of interest for applied microeconometricians since the 1950s, with at least 1,120 estimates in 139 countries \citep{psacharopoulos2018}. The formal modelling of human capital formation is grounded in microeconomic theory, which nearly any undergraduate Economics student will have come across \citep{schultz1960, schultz1961, mincer1958, becker1964}. Early contributions to this literature emphasised the heterogeneity of returns to education across levels of education and across individuals with the same level of education \citep{becker1966, chiswick1974, mincer1974}. A nonlinear model naturally accommodates this type of causal effect heterogeneity. A well-known confounder of the effect of education on earnings is ability \citep{griliches1977}. The dominant solution to the ability bias problem is the use of instruments for education \citep{heckman2006}. Parental education and number of siblings \citep{willis1979, taber2001} are poor instruments, as family background strongly determines ability \citep{cunha2006ea}. Another popular type of instrument is geographic location, e.g. distance to college \citep{card1993, kling2001, cameron2004}. However, the correlation of distance to college and an ability proxy casts doubt on its exogeneity \citep{carneiro2002}. Tuition cost \citep{kane1993} may be an invalid instrument due to its correlation with college quality, which may have a direct effect on earnings \citep{carneiro2002}. Local business cycle fluctuations were used as instruments in other studies \citep{cameron1998, carneiro2003, cameron2004}, but require that all permanent labour market effects are conditioned out, which may be difficult. The well-known quarter-of-birth instrument \citep{ak1991} is likely exogenous, but is known to be weak \citep{staiger1997}. Finding exogenous instruments is hard. Once an exogenous instrument is found, it is often weak or affects the treatment only in a subpopulation. In this case, the IV estimate can at best be interpreted as a marginal treatment effect in a subpopulation \citep{heckman2006}.

Economic theory states that ability is a likely confounder of the effect of education on earnings. Despite the previous long list of instruments, which intend to circumvent ability bias, ability proxies tend to explain little about the residuals in a regression of earnings on education and observable characteristics \citep{cawley1995}. Other unobserved characteristics like (non-academic) attitude \citep{green1998} and communication skills \citep{national1995} may be important determinants of earnings. 
Better communication skills serve applicants in admissions interviews and are highly valued in the labour market.

More generally, selection bias is inevitable when treatment is the result of heterogeneous individuals' utility-maximising choice.
Education is at least in part the result of individuals' expected utility $u_i$ optimisation, where the individuals' information set is naturally larger than that of the researcher. 
\begin{align}
Y_i &= g(A_i, U_i, \varepsilon_i) \\
A_i &= \max_{a \in \mathcal{A} \text{ s.t. constraints}} \left\{ \E{u_i(Y_i, A_i) | \text{ information set pre-college}} \right\} = h(Z_i, U_i, \eta_i)
\end{align}
Mechanically, individuals' optimising behaviour implies that in this model the disturbances $\varepsilon_i$ and $\eta_i$ are closely related to each other. Even conditional on the common confounder ability $U_i$, this dependence persists. However, there may be some instruments $Z$ relevant for education $A$ that satisfy an exclusion restriction conditional on ability $U_i$.


To demonstrate how instrumented common confounding suits the returns to education identification problems, we present a directed acyclic graph (DAG) to describe the model's conditional independence structures in figure \ref{f:ex2}. By going through the ICC construction algorithm, we demonstrate how the returns to education identification problem fits into the ICC framework. In terms of conditional independence assumptions we attempt to be as conservative as possible, which at worst leads to stronger than necessary richness requirements for the instruments $Z$ and proximal learning outcomes $W$.

\begin{figure} 
\caption{DAG of the Education Production Function}
\label{f:ex2}
\centering
\begin{minipage}{0.45 \textwidth}
\begin{tikzpicture}[node distance=1.5cm and 0.75cm]
    \node [state, dashed] (u) {$U$};
    \node [state, below left=1.5cm and 1.25cm of u] (a) {$A$};
    \node [state, below right=1.5cm and 1.25cm of u] (y) {$Y$};
    \node [draw, rounded rectangle, left=of a,  minimum width = 0.7 cm,  minimum height = 0.7 cm] (z) {$ \ \ \ Z \ \ \ $};
    \node [draw, rounded rectangle, right=of y, minimum width = 0.7 cm, minimum height = 0.7 cm] (w) {$W$};
    \node [state, dashed, below right=1.5cm and 1.25cm of a] (ut) {$\tilde{U}$};
    \node [cross=8pt, line width=4pt, a2red, below right=0.4cm and -0.4cm of z] (cross) {};
    \draw[line, style=-latex] (a) edge (y);
    \draw[line, style=-latex, a3sand, line width=2] (z) edge (a);
    \draw[line, style=-latex] (w) edge (y);
    \draw[line, a2red, densely dotted, out=270, in=225] (z) edge (y);
    \draw[line, a2red, densely dotted, out=270, in=180] (z) edge (ut);
    \draw[line, style=-latex, dashed] (u) edge (a);
    \draw[line, style=-latex, dashed] (u) edge (y);
    \draw[line, style=-latex, dashed] (ut) edge (a);
    \draw[line, style=-latex, dashed] (ut) edge (y);
    \draw[line, style=-latex, dashed, out=180, in=90] (u) edge (z);
    \draw[line, style=-latex, a4green, line width=2, dashed, out=0, in=90] (u) edge (w);
\end{tikzpicture}
\end{minipage}
\hfill
\begin{minipage}{0.5\textwidth}
\begin{itemize}
  \item[$Y$: ] Income 5 years post-college
  \item[$A$: ] College GPA
  \item[$Z$: ] GPA (high school), SAT, ACT
  \item[$U$: ] Ability, attitude, communication skills, general health, addictive tendency
  \item[$\tilde{U}$: ] General selection bias
  \item[$W$: ] Pre-college physical and behavioural characteristics
  \item[$X$: ] Family characteristics: parental education, income, occupation
\end{itemize}
\end{minipage}
\end{figure}
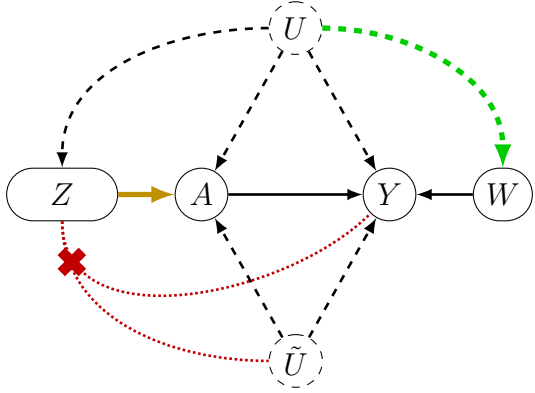

\begin{enumerate}
  \item $Z$ exogeneity (\ref{a:icc}.\ref{a:iv-exog}): Define common confounders $U$ such that \begin{align*} Y(a, z) = Y(a) \CI Z \ | \ U \ \forall a \in \mathcal{A}. \end{align*}
  Having considered the forms of confounding from ability $U$ and general selection bias $\tilde{U}$ in this model, the first task is to motivate the exclusion restriction conditional on ability $U$ for instruments $Z$. 
  As instrument $Z$ we use pre-college test scores and GPAs. Clearly, pre-college scores $Z$ are strongly associated with ability $U$. The main identifying assumption is that pre-college test scores and GPA measures do not affect post-college earnings through anything else than ability $U$ and college GPA $A$. This is a relatively sensible assumption. Once college GPA $A$ is determined, there is no reason why pre-college GPA $Z$ would matter regarding an individual's realised earnings post-college $Y$, other than through ability or other (family) background characteristics, which we assume to observe. 
  \item $W$ exogeneity (\ref{a:icc}.\ref{a:nc-exog}): Include in the common confounders $U$ any unobserved variables necessary to justify \begin{align*} W(a, z) = W \CI (A, Z) \ | \ U. \end{align*}
  As outcome-inducing proxies $W$, we use a range of pre-college physical and behavioural characteristics. These can include health status and behaviours, height, fitness, and addictive behaviours like alcoholism and smoking. While there may be some direct effects between pre-college health/behaviours and test scores, we can let $U$ capture all relevant information to ensure the required conditional independence. 
  \item $W$ relevance (\ref{a:icc}.\ref{a:nc-relevance}): Check whether $W$ is complete with respect to $U$ given $A$: \begin{align*} \E{g(A,U)|A,W} = 0 \text{ only when } g(A,U)=0 \text{ for any } g \in L_2(A,U). \end{align*}
There is substantial evidence for the relevance of the outcome-inducing proxies $W$ including physical characteristics and behaviours for ability \citep{gottfredson2004, case2008, tabriz2015, greengross2011}. Even allowing for some direct effects between academic performance and health/behavioural characteristics, their relevance for ability $U$ arguably remains likely.
Ultimately, it depends on the richness of available data whether completeness with respect to $U$ is satisfied. In some datasets, like the NLS97 \citep{nls97}, there are i.a. rich observations of health and addictive behaviour.
  \item $Z$ relevance (\ref{a:icc}.\ref{a:iv-complete}): Check whether $Z$ is complete with respect to $(A,U)$: \begin{align*} \E{g(A,U)|Z} = 0 \text{ only when } g(A,U)=0 \text{ for any } g \in L_2(A,U). \end{align*}
Pre-college scores $Z$ are strong determinants of college GPA $A$, even once ability $U$ is held fixed. Pre-college test scores can never measure ability perfectly, and also reflect some exam performance skills, which are likely carried on to college. The variation in $Z$ used to infer the effect of college GPA $A$ on post-college earnings $Y$ would be the excluded, non-ability determined variation in pre-college test scores.
Again, satisfaction of the relevance requirement for the instruments with respect to the treatment relies on the available data. In some datasets, like the NLS97 \citep{nls97}, there are rich observations of pre-college test scores $Z$.
\end{enumerate}

As long as the observational data is sufficiently rich, ICC can nonparametrically identify causal effects without too limiting exclusion restrictions. Any argument against the model's conditional exclusion restriction can be countered by including the hypothesised unobserved confounder into $U$. Clearly, at some point completeness of $Z$ with respect to $(A,U)$ and of $W$ with respect to $U$ is no longer a reasonable assumption as the number of hypothesised confounders keeps growing. In this sense, ICC has to trade off relevance against exclusion requirements, similar to standard IV. Unlike the rigid exclusion restriction in standard IV, ICC provides a way to utilise rich observational data when instruments are at best excluded conditional on some common confounders $U$. A good theoretical argument for ICC will consist of a sensible set of hypothesised common confounders $U$ and a convincing argument about the completeness of $Z$ and $W$ with respect to those (and of $Z$ with respect to $A$ given this set $U$). 

Ability or personality traits are optimal examples of such common confounders. Hence, ICC is suitable for the identification of various causal effects involving individuals' choices in rich observational data. Returns to education is, in our opinion, a good example of common confounding.

\subsection{Causal Effect of a Health Treatment} \label{ssec:ex-health}
Our second example involves the causal effect of a health treatment, which could be a drug, on a health outcome. For drugs, randomised controlled trials (RCTs) are the most common identification approach. Their advantages in overcoming potential unobserved confounding are obvious. However, as RCTs are costly, their sample size is mostly small and rather often effects of a drug cannot be tested rigorously for different subgroups of the population. After the introduction of a drug, the potential sample size hugely increases with treated individuals. With the ICC approach, we could estimate the causal effect of a treatment on a health outcome in observational data despite the presence of unobserved confounders. Again, we provide a DAG and explain the observed and unobserved variables we consider in this model in figure \ref{f:ex-health}. Then, we go through the ICC construction algorithm to explain how the health treatment problem fits the ICC model.

\begin{figure} 
\caption{DAG of a Health Treatment}
\label{f:ex-health}
\centering
\begin{minipage}{0.45 \textwidth}
\begin{tikzpicture}[node distance=1.5cm and 0.75cm]
    \node [state, dashed] (u) {$U$};
    \node [state, below left=1.5cm and 1.25cm of u] (a) {$A$};
    \node [state, below right=1.5cm and 1.25cm of u] (y) {$Y$};
    \node [draw, rounded rectangle, left=of a,  minimum width = 0.7 cm,  minimum height = 0.7 cm] (z) {$ \ \ \ Z \ \ \ $};
    \node [draw, rounded rectangle, right=of y, minimum width = 0.7 cm, minimum height = 0.7 cm] (w) {$W$};
    \node [state, dashed, below right=1.5cm and 1.25cm of a] (ut) {$\tilde{U}$};
    \node [cross=8pt, line width=4pt, a2red, below right=0.4cm and -0.4cm of z] (cross) {};
    \draw[line, style=-latex] (a) edge (y);
    \draw[line, style=-latex, a3sand, line width=2] (z) edge (a);
    \draw[line, style=-latex] (w) edge (y);
    \draw[line, a2red, densely dotted, out=270, in=225] (z) edge (y);
    \draw[line, a2red, densely dotted, out=270, in=180] (z) edge (ut);
    \draw[line, style=-latex, dashed] (u) edge (a);
    \draw[line, style=-latex, dashed] (u) edge (y);
    \draw[line, style=-latex, dashed] (ut) edge (a);
    \draw[line, style=-latex, dashed] (ut) edge (y);
    \draw[line, style=-latex, dashed, out=180, in=90] (u) edge (z);
    \draw[line, style=-latex, a4green, line width=2, dashed, out=0, in=90] (u) edge (w);
\end{tikzpicture}
\end{minipage}
\hfill
\begin{minipage}{0.5\textwidth}
\begin{itemize}
  \item[$Y$: ] Post-treatment health outcome
  \item[$A$: ] Health treatment
  \item[$Z$: ] Health measures pre-treatment \emph{with} effect on treatment choice/assignment
  \item[$U$: ] General health status
  \item[$\tilde{U}$: ] Belief in treatment sucess (placebo effect)
  \item[$W$: ] Health measures pre-treatment \emph{without} effect on treatment choice/assignment
\end{itemize}
\end{minipage}
\end{figure}
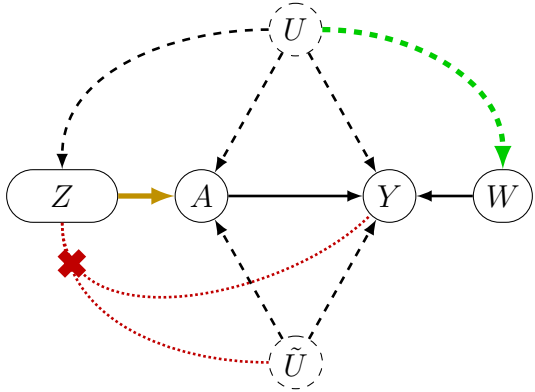

\begin{enumerate}
  \item $Z$ exogeneity (\ref{a:icc}.\ref{a:iv-exog}): Define common confounders $U$ such that \begin{align*} Y(a, z) = Y(a) \CI Z \ | \ U \ \forall a \in \mathcal{A}. \end{align*}
  Instruments $Z$ are pre-treatment health measures with an effect on treatment $A$ choice/assignment. Some of these instruments may even affect the post-treatment health outcome $Y$ directly. These instruments would form part of the unobserved common confounder $U$, which otherwise contains an individual's general health status. 
  Individuals with generally better health status and knowledge are more likely to choose treatment and might benefit more (or less) from the treatment. E.g. if treatment is a protein supplement and muscle growth the outcome variable, people with generally better health are likely to exercise more and hence likely to experience more muscle growth than their peers. 
  Another obvious confounder is the placebo effect, where individuals with greater belief in treatment success always benefit more from the treatment. This inherent bias problem is captured by the unobserved confounder $\tilde{U}$.
  \item $W$ exogeneity (\ref{a:icc}.\ref{a:nc-exog}): Include in the common confounders $U$ any unobserved variables necessary to justify \begin{align*} W(a, z) = W \CI (A, Z) \ | \ U. \end{align*}
$W$ contains pre-treatment health measures without effect on treatment choice/assignment. $W$ will be associated with other pre-treatment health measures via the general health status and some specific health issues which are measures by similar variables in $W$ and $Z$, which are consequently contained in $U$. 
  \item $W$ relevance (\ref{a:icc}.\ref{a:nc-relevance}): Check whether $W$ is complete with respect to $U$ given $A$: \begin{align*} \E{g(A,U)|A,W} = 0 \text{ only when } g(A,U)=0 \text{ for any } g \in L_2(A,U). \end{align*}
$W$ is complete with respect to $U$ as long as the information in its health measures is sufficiently informative with respect to the general status of health, as well as those specific pre-treatment health measures in $Z$ which affect $Y$. When health records of each individual are rich, this may well be the case.
  \item $Z$ relevance (\ref{a:icc}.\ref{a:iv-complete}): Check whether $Z$ is complete with respect to $(A,U)$: \begin{align*} \E{g(A,U)|Z} = 0 \text{ only when } g(A,U)=0 \text{ for any } g \in L_2(A,U). \end{align*}
$Z$ is complete with respect to $U$ by construction. Conditional on $U$, these pre-treatment health measures must be sufficiently relevant for treatment $A$. This may be quite likely when health records of each individual are rich.
\end{enumerate}

Importantly, the above example again shows that instruments $Z$ can affect outcome $Y$ to a significant degree without violating the ICC assumptions, simply by including those variables in $U$, and imposing richness requirements on the outcome-inducing proxies $W$. For example, let us consider a $J$-dimensional vector of relevant instruments $Z$ for the one-dimensional treatment $A$. Then, $J-1$ instruments may have a direct effect on outcome $Y$, as long as a $(J-1)$-dimensional vector of outcome-inducing proxies $W$ is complete for the $J-1$ endogenous instruments and excluded with respect to the one exogenous instrument and treatment $A$ conditional on the $J-1$ endogenous instruments. Hence, ICC can be used to obtain unbiased estimates even when some instruments are endogenous. It does not matter which instruments these are exactly, but completeness of $W$ with respect to the endogenous instruments is required, as well as the exclusion of $W$ with respect to the treatment $A$ and any exogenous instruments conditional on the endogenous ones. This flexibility of the ICC approach is one of its main strengths in rich observational data.

\section{Concluding Remarks} \label{sec:conclusion}
Instrumented common confounding (ICC) bridges identification theory between IV and proximal learning. The approach is suitable for causal inference in the social sciences, where ability or other character traits are common confounders. However, ICC is versatile and can be a viable solution to identification problems where instrument exclusion fails for a variety of different reasons.

Our two examples, one economic and one medical, demonstrate how we may observe instruments $Z$, which satisfy an exclusion restriction conditional on an unobserved common confounder $U$, for which other relevant and conditionally independent outcome-inducing proxies $W$ are available. In both examples, a theoretical argument in favour of traditional IV exclusion is too difficult. In ICC, we relax the exclusion restriction to exclusion conditional on an unobserved common confounder $U$. The proposed ICC model construction algorithm allows researchers to argue about the validity of ICC assumptions straightforwardly. Exclusion conditional on an unobserved common confounder $U$ is traded off with relevance of the observed variables $Z$ and $W$ for the unobserved common confounder $U$ and treatment $A$. 

While we are convinced that instrumented common confounding can be a useful tool for causal inference, it is no panacea. It replaces some strong, untestable identifying assumptions by other such assumptions, which are more likely to be satisfied in rich observational data. The validity of its assumptions must be justified rigorously by (economic) theory. 


In this paper, we dealt only with the identification of ICC models, not their estimation. While estimation is similar to standard IV, we intend to provide a dedicated package for the simple estimation of  ICC models using different statistical learning tools and look forward to applying the ICC method to various interesting empirical problems. We currently extend the ICC approach to allow for the conditional dependence of outcome-aligned proxies $W$ and treatment $A$.


\clearpage

\bibliographystyle{plainnat}
\bibliography{../references}

\clearpage

\section{Proofs} \label{sec:proofs}

\subsection{Proofs with outcome model restrictions} \label{ssec:proofs-outcome}

\begin{proof}[Proof of lemma \ref{l:cf-outcome}]


Let $T_A(x) \coloneqq \int_{\mathcal{A}} x(a^\prime) \pi(a^\prime) \dif \mu_A(a^\prime)$.
\begin{align*}
J &= \E{ \int_{\mathcal{A}} Y(a^\prime) \pi(a^\prime) \dif \mu_A(a^\prime)} = \E{ \int_{\mathcal{A}} \left( k_0(a^\prime, U) + \varepsilon_Y \right) \pi(a^\prime) \dif \mu_A(a^\prime)} \\
&= \E{ T_A \left( k_0(a^\prime, U) + \E{\varepsilon_Y | A=a^\prime, U} \right) } = \E{ T_A \left( k_0(a^\prime, U) + \E{\varepsilon_Y | U} \right)} \\
&= \E{ T_A \left( k_0(a^\prime, U) + \E{\E{\varepsilon_Y | Z, U} | U} \right) } =  \E{ T_A \left( k_0(a^\prime, U) \right) } \\
&= \E{ \left( \int_{\mathcal{A}} k_0(a^\prime, U) \pi(a^\prime) \dif \mu_A(a^\prime) \right)} = \E{\phi_{IV}(U; k_0)}
\end{align*}
The second line follows as for any change of $a$ in $Y(a)$, $\epsilon_Y$ is unchanged by definition. The third line follows as $\E{\varepsilon_Y | Z, U} = 0$.
\end{proof}

\begin{proof}[Proof of lemma \ref{l:h-outcome}]
\begin{align*}
\E{ \int_{\mathcal{A}} h_0(a, W) \pi(a) \dif \mu_A(a) } &=\E{ \E{ \int_{\mathcal{A}} h_0(a, W) \pi(a) \dif \mu_A(a) | U } } \\
&=\E{ \E{ \int_{\mathcal{A}} h_0(a, W) \pi(a) \dif \mu_A(a) | A=a, U } } \\
&=\E{ \int_{\mathcal{A}} \E{h_0(a, W) | A=a, U } \pi(a) \dif \mu_A(a) } \\
&=\E{ \int_{\mathcal{A}} k_0(a, U) \pi(a) \dif \mu_A(a) } \\
&=\E{\phi_{IV}(U; k_0)} = J
\end{align*}
From line one to two we used $W \CI A \ | \ U$ (A\ref{a:icc}.\ref{a:nc-exog}). From line three to four we used the definition of $\mathbb{H}_0$, where $h_0 \in \mathbb{H}_0$. On the last line we used lemma \ref{l:cf-outcome}.
\end{proof}

\begin{proof}[Proof of lemma \ref{l:obs-npiv-outcome}]

Any $h_0 \in \mathbb{H}_0$ satisfies
\begin{align*}
\E{k_0(A, U) - h_0(A, W) | A, U, Z} &= \E{k_0(A, U) - h_0(A, W) | A, U} = 0.
\end{align*}
The first equality holds 
by $W \CI Z \ | \ (A, U)$ (assumption \ref{a:icc}.\ref{a:nc-exog}). Consequently,
\begin{align*}
\E{Y - h_0(A, W) | Z} &= \E{ \E{k_0(A, U) - h_0(A, W) | U, Z} | Z} + \E{ \E{\varepsilon | U, Z} | Z} \\
&= \E{ \E{k_0(A, U) - h_0(A, W) | U, Z} | Z} \\
&= \E{k_0(A, U)  - h_0(A, W) | Z} \\
&= \E{ \E{k_0(A, U) - h_0(A, W) | A, U, Z} | Z} \\
&= \E{ \E{k_0(A, U) - h_0(A, W) | A, U} | Z} \\
&= \E{ 0 | Z} = 0
\end{align*}
This proves that equation \ref{eq:obs-mom-outcome} of lemma \ref{l:obs-npiv-outcome} holds.

\end{proof}

\begin{proof}[Proof of lemma \ref{l:h-equal-outcome}]
For any $h_0 \in \mathbb{H}_0^{\text{obs}}$,
\begin{align*}
\E{Y - h_0(A, W) | Z} = \E{k_0(A, U) - h_0(A, W) | Z} = \E{ \E{ k_0(A, U) - h_0(A, W) | A, U } | Z } = 0.
\end{align*}
Under completeness assumption \ref{a:icc}.\ref{a:iv-complete}, the above can only be true if $\E{ k_0(A, U) - h_0(A, W) | A, U } = 0$. Hence, any $h_0 \in \mathbb{H}_0^{\text{obs}}$ also satisfies $h_0 \in \mathbb{H}_0$, which implies $\mathbb{H}_0^{\text{obs}} \subseteq \mathbb{H}_0$. From lemma \ref{l:obs-npiv-outcome} it is known that $\mathbb{H}_0 \subseteq \mathbb{H}_0^{\text{obs}}$. Consequently, $\mathbb{H}_0^{\text{obs}} = \mathbb{H}_0$.
\end{proof}

\begin{proof}[Proof of lemma \ref{t:main-npiv-outcome}]

Identification of $\tilde{\phi}_{IV}$. For any $h_0 \in \mathbb{H}_0^{\text{obs}} = \mathbb{H}_0$,
\begin{align*}
\E{\tilde{\phi}_{IV}(W; h_0)} &= \E{\int_{\mathcal{A}} h_0(a, W) \pi(a) \dif \mu_A(a)} \\
&= \E{\int_{\mathcal{A}} \E{h(a, W) | A=a, U} \pi(a) \dif \mu_A(a)} \\
&= \E{\int_{\mathcal{A}} k_0(a, U) \pi(a) \dif \mu_A(a)} \\
&= \E{\phi_{IV}(U; k_0)} = J
\end{align*}
We move from the first to the second equation by assumption $W \CI (A, Z) \ | \ U$ of \ref{a:icc}.\ref{a:nc-exog}. The step from the second to the third line is by lemma \ref{l:h-equal-outcome} and \ref{a:icc}.\ref{a:nc-relevance}. The last line holds by lemma \ref{l:cf-outcome}.
\end{proof}

\subsection{Proofs with first stage restrictions} \label{ssec:proofs-first}

\begin{proof}[Proof of lemma \ref{l:simple-cf}]

\textbf{IPW}.
Identification of $\phi_{IPW}$.
\begin{align*}
\E{\phi_{IPW}(Y, A, V, U; f_{A | V, U})} &= \E{Y \frac{\pi(A)}{f_{A | V, U}(A | V, U)}} \\
 &= \E{Y \frac{\pi(A)}{f_{A | \eta, U}(A | \eta, U)}} \\
  &= \E{g(A, U, \varepsilon) \frac{\pi(A)}{f_{A | \eta, U}(A | \eta, U)}} \\
 &= \E{\int_{\mathcal{E}} \int_{\mathcal{A}} g(a, U, e) \frac{\pi(a)}{f_{A | \eta, U}(a | \eta, U)} f_{A, \varepsilon | \eta, U}(a, e | \eta, U) \dif \mu_A(a) \dif \mu_{\varepsilon}(e)} \\
 &= \E{\int_{\mathcal{E}} \int_{\mathcal{A}} g(a, U, e) \frac{\pi(a)}{f_{A | \eta, U}(a | \eta, U)} f_{A | \eta, U}(a | \eta, U) \dif \mu_A(a) f_{\varepsilon | \eta, U}(e | \eta, U) \dif \mu_{\varepsilon}(e)} \\
 &= \E{\int_{\mathcal{A}} g(a, U, \varepsilon) \pi(a) \dif \mu_A(a)} \\
&= \E{ \int_{\mathcal{A}} Y(a) \pi(a) \dif \mu_A(a)} = J
\end{align*}
On line one, we use assumption \ref{a:common-support-simple}. Line two uses the one-to-one correspondence of $V$ and $\eta$ conditional on $U$, which follows from assumption \ref{a:strict-monotonicity-simple}. Line five then uses assumption \ref{a:strict-monotonicity-simple}.\ref{a:mon-ind-simple}, which implies the independence of $A$ and $\varepsilon$ conditional on $(\eta, U)$.

\textbf{REG}.
Identification of $\phi_{REG}$.
\begin{align*}
\E{\phi_{REG}(V, U; k_{0, v})} &= \E{\int_{\mathcal{A}} k_{0, v}(a, V, U) \pi(a) \dif \mu_A(a)} \\
&= \E{\int_{\mathcal{A}} k_{0, \eta}(a, \eta, U) \pi(a) \dif \mu_A(a)} \\
&= \E{ \int_{\mathcal{A}} \E{g(A, U, \varepsilon) | A=a, \eta, U} \pi(a) \dif \mu_A(a)} \\
&= \E{ \int_{\mathcal{A}} \int_\mathcal{E} g(A, U, e) f_{\varepsilon | A, \eta, U}(e | a, \eta, U) \dif \mu_{\varepsilon}(e) \pi(a) \dif \mu_A(a)} \\
&= \E{ \int_{\mathcal{A}} \int_\mathcal{E} g(A, U, e) f_{\varepsilon | \eta, U}(e | \eta, U) \dif \mu_{\varepsilon}(e) \pi(a) \dif \mu_A(a)} \\
&= \E{ \int_\mathcal{E} \int_{\mathcal{A}} g(A, U, e) \dif \mu_\varepsilon(e) \pi(a) \dif \mu_A(a) f_{\varepsilon | \eta, U}(e | \eta, U) } \\
&= \E{ \int_{\mathcal{A}} g(a, U, \varepsilon) \pi(a) \dif \mu_A(a)} \\
&= \E{ \int_{\mathcal{A}} Y(a) \pi(a) \dif \mu_A(a)} = J
\end{align*}
On line one, we use assumption \ref{a:common-support-simple}. Line two uses the one-to-one correspondence of $V$ and $\eta$ conditional on $U$, which follows from assumption \ref{a:strict-monotonicity-simple}. Line five then uses assumption \ref{a:strict-monotonicity-simple}.\ref{a:mon-ind-simple}, which implies the independence of $A$ and $\varepsilon$ conditional on $(\eta, U)$.

\textbf{DR}.
Identification of $\phi_{DR}$.
\begin{align*}
\E{\phi_{DR}(Y, A, V, U; f_{A | V, U}, k_{0, v})} = \E{\phi_{REG}(V, U; k_{0, v})} = J
\end{align*}
\end{proof}

\begin{proof}[Proof of lemma \ref{l:same-info}]
Key to this proof is the strict monotonicity assumption \ref{a:strict-monotonicity}. Let $h^{-1}(., Z)$ be the inverse of $h(Z, m)$ in its second argument and $m^{-1}(., U)$ be the inverse of $m(U, \eta)$ in its second argument. Importantly, \ref{a:strict-monotonicity}.\ref{a:mon-2} implies that $m^{-1}(m, U)$ is strictly monotonous in its first argument. As in the proof of lemma 1 in \cite{matzkin2003}, we write the control function as
\begin{align*}
V_{\ref{a:strict-monotonicity}} &= \int_{\mathcal{U}} F_{A | Z, U}(A, Z, u) \dif F(u) \\
&= \int_{\mathcal{U}} \Pr{ \left( h(Z, m(u, \eta)) \leq A \right)} \dif F(u) \\
&= \int_{\mathcal{U}} \Pr{ \left( \eta \leq m^{-1} \left(h^{-1}(A, Z), u \right) \right) } \dif F(u) \\
&= \int_{\mathcal{U}} F_{\eta | U} \left( m^{-1} \left(h^{-1}(A, Z), u \right) \right)  \dif F(u) 
\end{align*}
Strict monotonocity of $m^{-1}(., U)$ in its first argument,
\begin{align*}
m^{-1}(m_1, U) > m^{-1}(m_0, U) \ \forall m_1 > m_0, \forall U \in \mathcal{U},
\end{align*}
implies no-crossing:

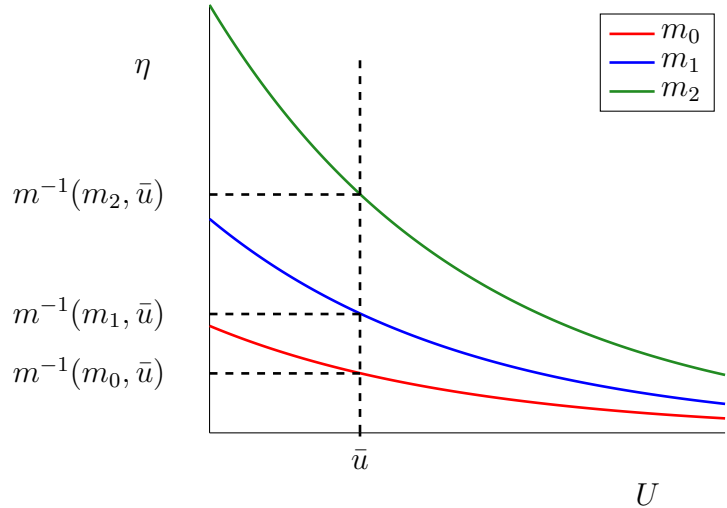
\begin{center}
\begin{tikzpicture}
  \begin{axis}[ 
    axis lines = left,
    xlabel = $U$,
    xlabel style={at={(axis description cs:0.85,0)}},
    ylabel = $\eta$,
    ylabel style={at={(axis description cs:0.05,0.85)},rotate=-90},
    ymin=0,
    ticks=none,
  ] 
    \addplot [ 
    line,
    domain=-1:1, 
    samples=100, 
    color=red, ]
    {0.5*exp(-x)}; 
\addlegendentry{$m_0$}
    \addplot [ 
    line,
    domain=-1:1, 
    samples=100, 
    color=blue, ]
    {exp(-x)};
\addlegendentry{$m_1$}
    \addplot [ 
    line,
    domain=-1:1, 
    samples=100, 
    color=ForestGreen, ]
    {2*exp(-x)};
\addlegendentry{$m_2$}
  \end{axis}
  \draw [line, dashed] (2,0) -- (2,5);
  \draw [line, dashed] (0, 3.17) -- (2, 3.17);
  \node at (-1.6, 3.17) {$m^{-1} ( m_2, \bar{u} )$};
  \draw [line, dashed] (0, 1.58) -- (2, 1.58);
  \node at (-1.6, 1.58) {$m^{-1} ( m_1, \bar{u} )$};
  \draw [line, dashed] (0, 0.79) -- (2, 0.79);
  \node at (-1.6, 0.79) {$m^{-1} ( m_0, \bar{u} )$};
  \draw [line, solid] (2,-0.05) -- (2,0);
  \node at (2, -0.35) {$\bar{u}$};
\end{tikzpicture}
\end{center}
Hence, the function $\eta(u) = m^{-1} \left(h^{-1}(a, z), u \right)$ is uniquely identified by its mean
\begin{align*}
\bar{M} = \int_{u} m^{-1} \left(h^{-1}(A, Z), u \right) dF(u).
\end{align*}
This implies the sigma algebras associated with the following three vectors of random variables are identical: $(\bar{M}, U), (\eta(U), U), (\eta, U)$. 

$F_{\eta | U}(\eta)$ is strictly monotonous in $\eta$ on its support due to the continuous conditional distribution of $\eta$ (assumption \ref{a:strict-monotonicity} 3). In combination with strict monotonicity of $m^{-1}(., U)$ in its first argument, this implies that $F_{\eta | U}\left(m^{-1} \left(h^{-1}(a, z), u \right) \right)$ is also strictly monotonous in $h^{-1}(a, z)$. 

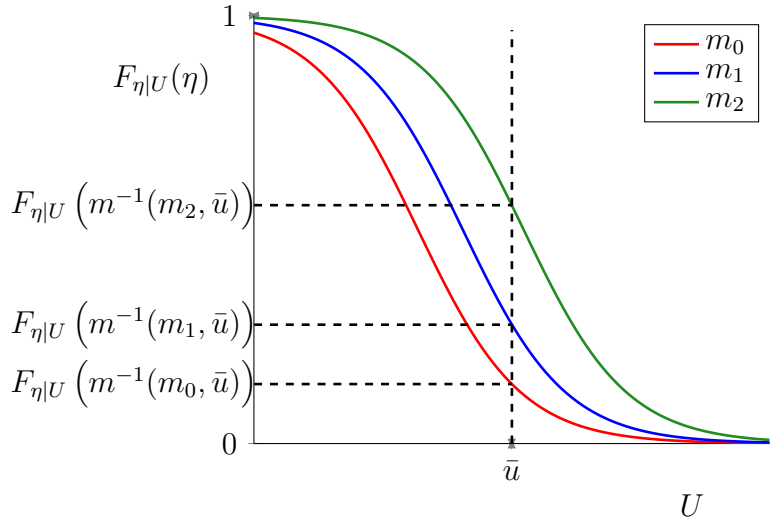
\begin{center}
\begin{tikzpicture}
  \begin{axis}[ 
    axis lines = left,
    xlabel = $U$,
    xlabel style={at={(axis description cs:0.85,0)}},
    ylabel = $F_{\eta | U}(\eta)$,
    ylabel style={at={(axis description cs:0,0.85)},rotate=-90},
    ymin = 0,
    ymax=1,
    ytick={0, 1},
    xtick={0},
    xticklabels={$\bar{u}$}
  ] 
    \addplot [ 
    line,
    samples=100, 
    color=red, ]
    {1-1/(1+exp(-x-1.82))}; 
\addlegendentry{$m_0$}
    \addplot [ 
    line,
    samples=100, 
    color=blue, ]
    {1-1/(1+exp(-x-0.95))};
\addlegendentry{$m_1$}
    \addplot [ 
    line,
    samples=100, 
    color=ForestGreen, ]
    {1-1/(1+exp(-x+0.23))};
\addlegendentry{$m_2$}
  \end{axis}
  \draw [line, dashed] (3.427,0) -- (3.427,5.5);
  \draw [line, dashed] (0, 3.17) -- (3.427, 3.17);
  \node at (-1.6, 3.17) {$ F_{\eta | U}\left(m^{-1} ( m_2, \bar{u} ) \right) $};
  \draw [line, dashed] (0, 1.58) -- (3.427, 1.58);
  \node at (-1.6, 1.58) {$ F_{\eta | U}\left(m^{-1} ( m_1, \bar{u} ) \right) $};
  \draw [line, dashed] (0, 0.79) -- (3.427, 0.79);
  \node at (-1.6, 0.79) {$ F_{\eta | U}\left(m^{-1} ( m_0, \bar{u} ) \right) $};
\end{tikzpicture}
\end{center}
Hence, the function $F_{\eta | U}(\eta(u)) = F_{\eta | U} \left(m^{-1} \left(h^{-1}(A, Z), u \right) \right)$ is uniquely identified by its mean
\begin{align*}
\bar{V} = \int_{u} F_{\eta | U} \left(m^{-1} \left(h^{-1}(A, Z), u \right) \right) dF(u).
\end{align*}
This implies the sigma algebras associated with the following vectors of random variables are identical: $(\bar{V}, U), (F_{\eta | U}(\eta(U)), U), (F_{\eta | U}(\eta), U), (\eta, U)$. The last two associated sigma algebras' equality follows from assumption \ref{a:strict-monotonicity}.\ref{a:mon-con} $F_{\eta | U}(\eta)$ is strictly monotonic on the support of $\eta$, the sigma algebra of $(F_{\eta | U}(\eta), U)$ and $(\eta, U)$ are equal. 
\end{proof}

\begin{proof}[Proof of lemma \ref{t:still-valid}]

First, note that the main step of both proofs for IPW and REG estimators, $\E{ \int_{\mathcal{A}} \E{Y | A=a, V_{\ref{a:strict-monotonicity}}, U} \pi(a) \dif \mu_A(a)} = \E{ \int_{\mathcal{A}} \E{Y | A=a, V, U} \pi(a) \dif \mu(a)}$, follows directly from lemma \ref{l:same-info}.

\textbf{IPW}. 
Want to show that $\E{\phi_{IPW}(Y, A, V_{\ref{a:strict-monotonicity}}, U; f_{A | V_{\ref{a:strict-monotonicity}}, U})} = \E{\phi_{IPW}(Y, A, V, U; f_{A | V, U})}$.
\begin{align*}
\E{\phi_{IPW}(Y, A, V_{\ref{a:strict-monotonicity}}, U; f_{A | V_{\ref{a:strict-monotonicity}}, U})} &= \E{Y \frac{\pi(A)}{f_{A | V_{\ref{a:strict-monotonicity}}, U}(A | V_{\ref{a:strict-monotonicity}}, U)}} \\
&= \E{ \int_{\mathcal{A}} \E{Y | A=a, V_{\ref{a:strict-monotonicity}}, U} \frac{\pi(a)}{f_{A | V_{\ref{a:strict-monotonicity}}, U}(a | V_{\ref{a:strict-monotonicity}}, U)} f_{A | V_{\ref{a:strict-monotonicity}}, U}(a | V_{\ref{a:strict-monotonicity}}, U) \dif \mu_A(a)} \\
&= \E{ \int_{\mathcal{A}} \E{Y | A=a, V_{\ref{a:strict-monotonicity}}, U} \pi(a) \dif \mu_A(a)} \\
&= \E{ \int_{\mathcal{A}} \E{Y | A=a, V, U} \pi(a) \dif \mu_A(a)} \\
&= \E{\phi_{IPW}(Y, A, V, U; f_{A | V, U})}
\end{align*}
The penultimate line follows from lemma \ref{l:same-info}, and the last line from lemma \ref{l:simple-cf}.

\textbf{REG}. 
Want to show that $\E{\phi_{REG}(V_{\ref{a:strict-monotonicity}}, U; k_{0, V_{\ref{a:strict-monotonicity}}})} = \E{\phi_{REG}(V, U; k_{0, v})}$.
\begin{align*}
\E{\phi_{REG}(V_{\ref{a:strict-monotonicity}}, U; k_{0, V_{\ref{a:strict-monotonicity}}})} &= \E{\int_{\mathcal{A}} k_0(a, V_{\ref{a:strict-monotonicity}}, U) \pi(a) \dif \mu_A(a)} \\
&= \E{ \int_{\mathcal{A}} \E{Y | A=a, V_{\ref{a:strict-monotonicity}}, U} \pi(a) \dif \mu_A(a)} \\
&= \E{ \int_{\mathcal{A}} \E{Y | A=a, V, U} \pi(a) \dif \mu_A(a)} \\
&= \E{\phi_{REG}(V, U; k_0)}
\end{align*}
The last line follows from the definition of $\phi_{REG}(.;.)$ in lemma \ref{l:simple-cf}.

\textbf{DR}. 
Want to show that $\E{\phi_{DR}(Y, A, V_{\ref{a:strict-monotonicity}}, U; f, k_0)} = \E{\phi_{DR}(Y, A, V, U; f, k_0)}$.
\begin{align*}
\E{\phi_{DR}(Y, A, V_{\ref{a:strict-monotonicity}}, U; f_{A | V_{\ref{a:strict-monotonicity}}, U}, k_{0, V_{\ref{a:strict-monotonicity}}})} &= \E{\phi_{REG}(V_{\ref{a:strict-monotonicity}}, U; k_{0, V_{\ref{a:strict-monotonicity}}})} \\
&= \E{\phi_{REG}(V, U; k_{0, v})} = \E{\phi_{DR}(Y, A, V, U; f_{A | V, U}, k_{0, v})}
\end{align*}
The penultimate line follows from lemma \ref{l:same-info}, and the last line from lemma \ref{l:simple-cf}.
\end{proof}

\begin{proof}[Proof of lemma \ref{l:control-obs}]
Proof of equation \ref{eq:control-obs-mom-1}. For any $\tau_{A,0} \in \mathbb{T}_0$,
\begin{align*}
\E{\tau_{A,0}(Z, W_1) | Z, U, W_0} &= F(A | Z, U) \\
\E{ \E{\tau_{A,0}(Z, W_1) | Z, U, W_0} | Z, W_0} &= \E{F(A | Z, U) | Z, W_0} \\
\E{ \tau_{A,0}(Z, W_1) | Z, W_0} &= F(A | Z, W_0).
\end{align*}
We move from the second to third equation by assumption $A \CI W_0 \ | \ (Z, U)$ (assumption \ref{a:icc-first-add}).

Proof of equation \ref{eq:control-obs-mom-2}. For any $\kappa_0 \in \mathbb{K}_0$,
\begin{align*}
\E{\kappa_0(Z, W_0) | Z, U, W_1} &= \frac{f(U)}{f(U | Z)} \\
\E{ \E{ \kappa_0(Z, W_0) | Z, U, W_1} | Z, W_1} &= \E{ \frac{f(U)}{f(U | Z)} | Z, W_1} \\
\E{ \E{ \kappa_0(Z, W_0) | Z, U, W_1} | Z, W_1} &= f(Z) \E{ \frac{1}{f(Z | U)} | Z, W_1} \\
\E{ \kappa_0(Z, W_0) | Z, W_1}
&= f(Z) \int_{\mathcal{U}} \frac{f(U | W_1, Z)}{f(Z | U)} \dif \mu_U(U) \\
&= f(Z) \int_{\mathcal{U}} \frac{f(W_1, Z | U) f(U)}{f(Z | U) f(W_1, Z)} \dif \mu_U(U) \\
&= f(Z) \int_{\mathcal{U}} \frac{f(W_1 | U) f(U)}{f(W_1, Z)} \dif \mu_U(U) \\
&= \frac{f(Z)}{f(Z | W_1)} \\
&= \frac{f(W_1)}{f(W_1 | Z)}
\end{align*}
We move from the fifth to the sixth equation by assumption $W_1 \CI Z \ | \ U$ (assumption \ref{a:icc-first-add}).
\end{proof}

\begin{proof}{Proof of lemma \ref{l:control-obs-id}}
First, note that for any $\tau_{A,0} \in \mathbb{T}_0$,
\begin{align*}
F(A | Z, U) &= \E{ \tau_{A,0}(Z, W_1) | Z, U, W_0 } \\
&= \E{ \tau_{A,0}(Z, W_1) | Z, U } \\
&= \int_{\mathcal{W}_1} \tau_{A,0}(Z, W_1) \dif F(W_1 | U)  \\
V_{\ref{a:strict-monotonicity}} = \int_{\mathcal{U}} F(A | Z, U) \dif F(U) &= \int_{\mathcal{W}_1} \tau_{A,0}(Z, W_1) \dif F(W_1) 
\end{align*}
We move from the second to the third line by assumption $W_1 \CI Z \ | \ U$ (assumption \ref{a:icc-first-add}).

Now, note that for any $\tau_A \in L_2(Z, W_1)$ and $\kappa_0 \in \mathbb{K}_0^{\text{obs}}$,
\begin{align*}
\E{ \kappa_0(Z, W_0) \tau_A(Z, W_1) | Z } &= \E{ \E{ \kappa_0(Z, W_0) \tau_A(Z, W_1) | Z, W_1 } | Z } \\
&= \E{ \E{ \kappa_0(Z, W_0) | Z, W_1 } \tau_A(Z, W_1) | Z } \\
&= \E{ \frac{f(W_1)}{f(W_1 | Z)} \tau_A(Z, W_1) | Z } \\
&= \int_{\mathcal{W}_1} \frac{f(W_1)}{f(W_1 | Z)} \tau_A(Z, W_1) \dif F(W_1 | Z) \\
&= \int_{\mathcal{W}_1} \tau_A(Z, W_1) \dif F(W_1)
\end{align*}

For any $\tau_A \in L_2(Z, W_1)$ and $\kappa_0 \in \mathbb{K}_0^{\text{obs}}$, we write
\begin{align*}
\int_{\mathcal{W}_1} \tau_A(Z, W_1) \dif F(W_1) &= \E{ \kappa_0(Z, W_0) \tau_A(Z, W_1) | Z } \\
&= \E{ \kappa_0(Z, W_0) \E{ \tau_A(Z, W_1) | Z, W_0} | Z }.
\end{align*}

For any $\tau_{A, 0} \in \mathbb{T}_0$ and $\kappa_0 \in \mathbb{K}_0^{\text{obs}}$, we write
\begin{align*}
V_{\ref{a:strict-monotonicity}} &= \int_{\mathcal{W}_1} \tau_{A, 0}(Z, W_1) \dif F(W_1) \\
&= \E{ \kappa_0(Z, W_0) \tau_{A, 0}(Z, W_1) | Z } \\
&= \E{ \kappa_0(Z, W_0) \E{ \tau_{A, 0}(Z, W_1) | Z, W_0} | Z } \\
&= \E{ \kappa_0(Z, W_0) F(A | Z, W_0) | Z }.
\end{align*}
Now this implies that for any $\tau_A \in L_2(Z, W_1)$,
\begin{align*}
\int_{\mathcal{W}_1} \tau_A(Z, W_1) \dif F(W_1) - V_{\ref{a:strict-monotonicity}} &= \E{ \kappa_0(Z, W_0) \E{ \tau_A(Z, W_1) - F(A | Z, W_0) | Z, W_0} | Z }.
\end{align*}
Hence, for any $\tau_A \in \mathbb{T}_0^{\text{obs}}$ as long as $\mathbb{K}_0^{\text{obs}} \neq \emptyset$,
\begin{align*}
V_{\ref{a:strict-monotonicity}} &= \int_{\mathcal{W}_1} \tau_A(Z, W_1) \dif F(W_1).
\end{align*}
\end{proof}

\begin{proof}[Proof of lemma \ref{l:main}]

\textbf{IPW}.
Identification of $\tilde{\phi}_{IPW}$. For any $q_0 \in \mathbb{Q}_0$,
\begin{align*}
\E{\tilde{\phi}_{IPW}(Y, A, \tilde{V}, Z)} &= \E{Y \pi(A) q_0(A, \tilde{V}, Z)} \\
&= \E{ \E{Y \pi(A) q_0(A, \tilde{V}, Z) | Y, A, \tilde{V}, U}} \\
&= \E{ Y \E{\pi(A) q_0(A, Z) | A, \tilde{V}, U}} \\
&= \E{ Y \frac{\pi(A)}{f(A | \tilde{V}, U)}} = \E{ {\phi}_{IPW}(Y, A, \tilde{V}, U) } \\
&= \E{ Y \frac{\pi(A)}{f(A | \eta, U)}} = J
\end{align*}
We move from the second to third equation by assumption $Y \CI Z \ | \ (A, U)$ of \ref{a:icc-first-add}. If \ref{a:strict-monotonicity} holds, the last line holds by theorem \ref{t:still-valid}. Otherwise, the last line directly holds by the steps in the proof of lemma \ref{l:simple-cf}, as \ref{a:strict-monotonicity-no-cf-fs} implies that $\tilde{V} = V_{\ref{a:strict-monotonicity-no-cf-fs}}$ is a one-to-one transformation of $\eta$.


\textbf{REG}.
Identification of $\tilde{\phi}_{REG}$. For any $h_0 \in \mathbb{H}_0$,
\begin{align*}
\E{\tilde{\phi}_{REG}(\tilde{V}, W)} &= \E{\int_{\mathcal{A}} h_0(a, \tilde{V}, W) \pi(a) \dif \mu_A(a)} \\
&= \E{\int_{\mathcal{A}} \E{h(a, \tilde{V}, W) | A=a, \tilde{V}, U} \pi(a) \dif \mu_A(a)} \\
&= \E{\int_{\mathcal{A}} k_0(a, \tilde{V}, U) \pi(a) \dif \mu_A(a)} = \E{\phi_{REG}(\tilde{V}, U)} \\
&= \E{\int_{\mathcal{A}} k_0(a, \eta, U) \pi(a) \dif \mu_A(a)} = J
\end{align*}
We move from the first to the second equation by assumption $W \CI (A, Z) \ | \ U$ of \ref{a:icc-first-add}. If \ref{a:strict-monotonicity} holds, the last line holds by theorem \ref{t:still-valid}. Otherwise, the last line directly holds by the steps in the proof of lemma \ref{l:simple-cf}, as \ref{a:strict-monotonicity-no-cf-fs} implies that $\tilde{V} = V_{\ref{a:strict-monotonicity-no-cf-fs}}$ is a one-to-one transformation of $\eta$.

\textbf{DR}.
Identification of $\tilde{\phi}_{DR}$.
\begin{align*}
\E{\tilde{\phi}_{DR}(Y, A, \tilde{V}, W, Z)} &= \E{ \pi(A) q_0(A, \tilde{V}, Z) (Y - h_0(A, \tilde{V}, W)) }  + \E{\tilde{\phi}_{REG}(Y, A, \tilde{V}, W)} \\
&= \E{ \E{ \pi(A) q_0(A, \tilde{V}, Z) (Y - h_0(A, \tilde{V}, W)) | A, \tilde{V}, U }}  + J \\
&= \E{ \E{ \pi(A) q_0(A, \tilde{V}, Z) | A, \tilde{V}, U} \E{ (Y - h_0(A, \tilde{V}, W)) | A, \tilde{V}, U }}  + J \\
&= \E{ \E{ \pi(A) q_0(A, \tilde{V}, Z) | A, \tilde{V}, U}  (\E{Y | A, \tilde{V}, U } - k_0(A, \tilde{V}, U)) }  + J \\
&= \E{ \E{ \pi(A) q_0(A, \tilde{V}, Z) | A, \tilde{V}, U}  (k_0(A, \tilde{V}, U) - k_0(A, \tilde{V}, U)) }  + J \\
&= J
\end{align*}
We move from the second to the third equation by assumption $(Y, W) \CI Z \ | \ (U, A)$ of \ref{a:icc-first-add}. 
\end{proof}

\begin{proof}[Proof of lemma \ref{l:obs}]
Any $h_0 \in \mathbb{H}_0$ satisfies
\begin{align*}
\E{Y - h_0(A, \tilde{V}, W) | A, \tilde{V}, Z, U} &= \E{Y - h_0(A, \tilde{V}, W) | A, \tilde{V}, U} = 0.
\end{align*}
The first equality holds 
by $(W, Y) \CI Z | A, U$. Consequently,
\begin{align*}
\E{Y - h_0(A, \tilde{V}, W) | A, \tilde{V}, Z} &= \E{ \E{ Y - h_0(A, \tilde{V}, W) | A, \tilde{V}, Z, U } | A, \tilde{V}, Z } = 0.
\end{align*}
This proves that equation \ref{eq:obs-mom-1} of lemma \ref{l:obs} holds.
Similarly, any $q_0 \in \mathbb{Q}_0$ satisfies
\begin{align*}
\E{\pi(A) q_0(A, \tilde{V}, Z) | A, \tilde{V}, W, U} &= \E{\pi(A) q_0(A, \tilde{V}, Z) | A, \tilde{V}, U} = \frac{\pi(A)}{f(A | \tilde{V}, U)}.
\end{align*}
The first equality holds 
by $Z \CI W \ | \ A, U$. Consequently,
\begin{align*}
\E{\pi(A) q_0(A, \tilde{V}, Z) | A, \tilde{V}, W} &= \E{ \E{\pi(A) q_0(A,\tilde{V}, Z) | A, \tilde{V}, U, W} | A, \tilde{V}, W } \\
&= \E{ \frac{\pi(A)}{f(A | \tilde{V}, U)} \Big| A, \tilde{V}, W }.
\end{align*}
Equation \ref{eq:obs-mom-2} of lemma \ref{l:obs} holds because
\begin{align*}
\E{\frac{1}{f(A | \tilde{V}, U)} | A, \tilde{V}, W} &= \int \frac{1}{f(A | \tilde{V}, u)} f(u | A, \tilde{V}, W) \dif \mu_U(u) \\
&= \int \frac{f(A, W | u, \tilde{V}) f(u, \tilde{V})}{f(A | \tilde{V}, u) f(A, \tilde{V}, W)} \dif \mu_U(u) \\
&= \int \frac{f(A | u, \tilde{V}) f(W | u, \tilde{V}) f(u, \tilde{V})}{f(A | \tilde{V}, u) f(A, \tilde{V}, W)} \dif \mu_U(u) \\
&= \int \frac{f(W | u, \tilde{V}) f(u, \tilde{V})}{f(A, \tilde{V}, W)} \dif \mu_U(u) \\
&= \frac{f(W, \tilde{V})}{f(A, \tilde{V}, W)} = \frac{1}{f(A | \tilde{V}, W)}
\end{align*}
The third equality follows from conditional independence $W \CI Z \ | \ U$ and $A=h(Z, \eta)$ under assumption \ref{a:strict-monotonicity-no-cf-fs}, or $A=h(Z, m(U, \eta))$ under assumption \ref{a:strict-monotonicity} (with \ref{a:icc-first-add},  \ref{a:control-bridge} to identify $\tilde{V}$). 
\end{proof}

\begin{proof}[Proof of lemma \ref{l:5}]

\textbf{IPW}.
First we prove equation of \ref{eq:l5-ipw} of lemma \ref{l:5} for the IPW estimator. For any $h_0 \in \mathbb{H}_0^{\text{obs}}$,
\begin{align*}
\E{\tilde{\phi}_{IPW}(Y, A, \tilde{V}, Z; q)} &= \E{ Y \pi(A) q(A, \tilde{V}, Z) } \\
&= \E{ \E{Y | A, \tilde{V}, Z} \pi(A) q(A, \tilde{V}, Z) } \\
&= \E{ \E{h_0(A, \tilde{V}, W) | A, \tilde{V}, Z} \pi(A) q(A, \tilde{V}, Z) } \\
&= \E{ h_0(A, \tilde{V}, W)  \pi(A) q(A, \tilde{V}, Z) } \\
&= \E{ h_0(A, \tilde{V}, W) \E{ \pi(A) q(A, \tilde{V}, Z) | A, \tilde{V}, W } }.
\end{align*}
The third line requires equation \ref{eq:obs-mom-1} from lemma \ref{l:obs}.
For any $q_0 \in \mathbb{Q}_0$,
\begin{align*}
\E{\frac{\pi(A)}{f(A | \tilde{V}, W)} h_0(A, \tilde{V}, W)} 
&= \E{ \E{\pi(A) q_0(A, \tilde{V}, Z) | A, \tilde{V}, W} h_0(A, \tilde{V}, W) } \\
&= \E{ \E{\pi(A) q_0(A, \tilde{V}, Z) | A, \tilde{V}, W} \E{ Y | A, \tilde{V}, W} } \\
&= \E{ \E{\pi(A) q_0(A, \tilde{V}, Z) Y | A, \tilde{V}, W} } \\
&= \E{ \pi(A) q_0(A, \tilde{V}, Z) Y  } = J  
\end{align*}
The first line requires \ref{eq:obs-mom-2} from lemma \ref{l:obs}.
Combining both above results, for any $h_0 \in \mathbb{H}_0^{\text{obs}}$ as long as $\mathbb{Q}_0 \neq \emptyset$, we have
\begin{align*}
\E{\tilde{\phi}_{IPW}(Y, A, \tilde{V}, Z; q)} - J &= \E{ \left( \E{ \pi(A) q(A, \tilde{V}, Z) | A, \tilde{V}, W} - \frac{\pi(A)}{f(A | \tilde{V}, W)} \right) h_0(A, \tilde{V}, W)}.
\end{align*}

\textbf{REG}.
Now we prove equation of \ref{eq:l5-reg} of lemma \ref{l:5} for the REG estimator. Again, we use equations \ref{eq:obs-mom-1} and \ref{eq:obs-mom-2} from lemma \ref{l:obs}. For any $q_0 \in \mathbb{Q}_0^{\text{obs}}$,
\begin{align*}
\E{\tilde{\phi}_{REG}(\tilde{V}, W; h)} &= \E{(\mathcal{T}h)(\tilde{V}, W)} \\
&= \E{ \int_{\mathcal{A}} \frac{\pi(a)}{f(a | \tilde{V}, W)} h(a, \tilde{V}, W) f(a | \tilde{V}, W) d \mu_A(a) } \\
&= \E{ \E{ \frac{\pi(A)}{f(A | \tilde{V}, W)} h(A, \tilde{V}, W) | \tilde{V}, W } } \\
&= \E{  \frac{\pi(A)}{f(A | \tilde{V}, W)} h(A, \tilde{V}, W)  } \\
&= \E{ \E{ \pi(A) q_0(A, \tilde{V}, Z) | A, \tilde{V}, W }  h(A, \tilde{V}, W)  } \\
&= \E{ \pi(A) q_0(A, \tilde{V}, Z)  h(A, \tilde{V}, W)  }.
\end{align*}
For any $h_0 \in \mathbb{H}_0$,
\begin{align*}
\E{\pi(A) q_0(A, \tilde{V}, Z)  \E{Y | A, \tilde{V}, Z}} &=  \E{\pi(A) q_0(A, \tilde{V}, Z)  \E{h_0(A, \tilde{V}, W) | A, \tilde{V}, Z}} \\
&= \E{\pi(A) q_0(A, \tilde{V}, Z)  h_0(A, \tilde{V}, W) } \\
&= \E{(\mathcal{T}h_0)(\tilde{V}, W)} = J. \text{ by lemma \ref{l:main}}
\end{align*}
The first line holds by lemma \ref{l:obs} equation \ref{eq:obs-mom-1}.
Combining both above results, for any $q_0 \in \mathbb{Q}_0^{\text{obs}}$ as long as $\mathbb{H}_0 \neq \emptyset$, we have 
\begin{align*}
\E{\tilde{\phi}_{REG}(\tilde{V}, W; h)} - J &= \E{ \pi(A) q_0(A, \tilde{V}, Z) \E{h(A, \tilde{V}, W) - Y | A, \tilde{V}, Z}}.
\end{align*}
\end{proof}

\begin{proof}[Proof of lemma \ref{t:4}]
\textbf{IPW}. First we prove equation \ref{eq:l4-ipw} of lemma \ref{t:4}. We use lemma \ref{l:5} with some $h_0 \in \mathbb{H}_0^{\text{obs}}$ and take a $q_0 \in \mathbb{Q}_0^{\text{obs}}$, such that
\begin{align*}
\E{\tilde{\phi}_{IPW}(Y, A, \tilde{V}, Z; q_0)} - J &= \E{ \left( \E{ \pi(A) q_0(A, \tilde{V}, Z) | A, \tilde{V}, W} - \frac{\pi(A)}{f(A | \tilde{V}, W)} \right) h_0(A, \tilde{V}, W)} = 0,
\end{align*}
by the definition of $\mathbb{Q}_0^{\text{obs}}$. This requires $\{ \mathbb{Q}_0 \neq \emptyset$ and $\mathbb{H}_0^{\text{obs}} \neq \emptyset \}$ as in lemma \ref{l:5}.

\textbf{REG}. Now we prove equation \ref{eq:l4-reg} of theorem \ref{t:4}. We use lemma \ref{l:5} with some $q_0 \in \mathbb{Q}_0^{\text{obs}}$ and take a $h_0 \in \mathbb{H}_0^{\text{obs}}$, such that
\begin{align*}
\E{\tilde{\phi}_{REG}(\tilde{V}, W; h_0)} - J &= \E{ \pi(A) q_0(A, Z) \E{h_0(A, \tilde{V}, W) - Y | A, \tilde{V}, Z}} = 0.
\end{align*}
by the definition of $\mathbb{H}_0^{\text{obs}}$. This requires $\{ \mathbb{H}_0 \neq \emptyset$ and $\mathbb{Q}_0^{\text{obs}} \neq \emptyset \}$ as in lemma \ref{l:5}.

\textbf{DR}. Double robustness is shown as usual. Suppose $\{ \mathbb{H}_0 \neq \emptyset$ and $\mathbb{Q}_0^{\text{obs}} \neq \emptyset \}$. For any $h_0 \in \mathbb{H}_0^{\text{obs}}$ and $q \in L_2(A, Z)$,
\begin{align*}
\E{\tilde{\phi}_{DR}(Y, A, \tilde{V}, W; h_0, q)} &= \E{(Y - h_0(A, \tilde{V}, W)) \pi(A) q(A, \tilde{V}, Z) + (\mathcal{T} h_0)(\tilde{V}, W)} \\
&= \E{ \E{(y - h_0(A, \tilde{V}, W)) | A, Z} \pi(A) q(A, \tilde{V}, Z) } + \E{(\mathcal{T} h_0)(\tilde{V}, W)} \\
&= \E{\tilde{\phi}_{REG}(\tilde{V}, W; h_0)} = J.
\end{align*}
Suppose $\{ \mathbb{Q}_0 \neq \emptyset$ and $\mathbb{H}_0^{\text{obs}} \neq \emptyset \}$. For any $q_0 \in \mathbb{Q}_0^{\text{obs}}$ and $h \in L_2(A, \tilde{V}, W)$,
\begin{align*}
&\E{\tilde{\phi}_{DR}(Y, A, \tilde{V}, W; h, q_0)} \\
&= \E{(Y - h(A, \tilde{V}, W)) \pi(A) q_0(A, \tilde{V}, Z) + (\mathcal{T} h)(\tilde{V}, W)} \\
&= \E{(Y - h(A, \tilde{V}, W)) \pi(A) q_0(A, \tilde{V}, Z) + \E{h(A, \tilde{V}, W) \frac{\pi(A)}{f(A | \tilde{V}, W)} | \tilde{V}, W}} \\
&= \E{Y \pi(A) q_0(A, \tilde{V}, Z) + h(A, \tilde{V}, W) \pi(A) \left( \frac{1}{f(A | \tilde{V}, W)} - q_0(A, \tilde{V}, Z) \right) } \\
&= \E{Y \pi(A) q_0(A, \tilde{V}, Z)} + \E{ h(A, \tilde{V}, W) \pi(A) \E{ \frac{1}{f(A | \tilde{V}, W)} - q_0(A, \tilde{V}, Z) | A, \tilde{V}, W } } \\
&= \E{Y \pi(A) q_0(A, \tilde{V}, Z) } \\
&= \E{\tilde{\phi}_{IPW}(Y, A, \tilde{V}, Z; q_0)} = J.
\end{align*}
\end{proof}

\end{document}